\documentclass[11pt]{article}

\usepackage[margin=1in]{geometry}
\usepackage{amsmath,amssymb,amsthm,mathtools,bm}
\usepackage{enumitem}
\usepackage{needspace}
\usepackage{tikz}
\usetikzlibrary{arrows.meta,positioning}
\usepackage{xcolor}
\usepackage{cite}
\usepackage{authblk}

\usepackage[
  colorlinks=true,
  linkcolor=blue!55!black,
  citecolor=blue!55!black,
  urlcolor=blue!55!black,
  pdftitle={A Weakly Group Strategy-Proof and Pareto-Stable Matching Mechanism with Two-Sided Indifferences under Gross Substitutes},
  pdfauthor={Yasushi Kawase and Ryosuke Sato}
]{hyperref}
\usepackage[nameinlink,noabbrev,sort&compress]{cleveref}
\crefname{appendix}{appendix}{appendices}
\Crefname{appendix}{Appendix}{Appendices}

\usepackage{titlesec}
\titlespacing*{\paragraph}{0pt}{.8\baselineskip}{1em} 
\usepackage[many]{tcolorbox}
\newtcolorbox{procbox}[1][]{%
  enhanced, sharp corners,
  colback=white, colframe=black, boxrule=0.5pt,
  left=1em, right=1em, top=0.7em, bottom=0.7em,
  #1}

\newtheorem{theorem}{Theorem}[section]
\newtheorem{proposition}[theorem]{Proposition}
\newtheorem{lemma}[theorem]{Lemma}
\newtheorem{corollary}[theorem]{Corollary}
\theoremstyle{definition}

\newtheorem{example}[theorem]{Example}

\newcommand{\Mnat}{M\ensuremath{{}^{\natural}}}
\newcommand{\cA}{\mathcal{A}}
\newcommand{\cC}{\mathcal{C}}
\newcommand{\cD}{\mathcal{D}}

\newcommand{\cI}{\mathcal{I}}

\newcommand{\cM}{\mathcal{M}}

\newcommand{\dom}{\operatorname{dom}}
\newcommand{\argmax}{\operatorname*{arg\,max}}
\renewcommand{\mid}{\mathrel{:}}

\title{A Weakly Group Strategy-Proof and Pareto-Stable Matching Mechanism with Two-Sided Indifferences under Gross Substitutes}
\author{Yasushi Kawase}
\author{Ryosuke Sato}
\affil{Chuo University}
\date{}

\begin{document}
\maketitle

\begin{abstract}
We study finite many-to-one matching markets with indifferences on both sides and publicly known firm preferences over sets of workers.
A standard approach breaks ties uniformly at random and runs deferred acceptance, which is strategy-proof for workers but may select a weakly stable matching that is Pareto dominated under the agents' original preferences.
For the broad class of markets in which each firm's preferences admit a gross substitutes (GS) representation, we construct a deterministic generalized deferred acceptance mechanism that is weakly stable, Pareto efficient with respect to both sides, worker weakly group strategy-proof, and implementable in strongly polynomial time given value oracles for the GS representations.
This domain includes nonresponsive firm preferences, such as additive scores under matroid constraints and assignment values.
Our analysis bridges ordinal matching with indifferences and competitive equilibria under imperfectly transferable utility.
Although the mechanism uses only ordinal comparisons, it returns a seller-optimal equilibrium matching of a tiered auxiliary market, and the equilibrium structure of this market yields the stability, efficiency, and incentive guarantees.
\end{abstract}

\section{Introduction}\label{sec:introduction}

Centralized markets that match workers to firms often face ties in preferences.
A standard approach breaks the ties before running worker-proposing deferred acceptance~\cite{AbdulkadirogluPathakRoth2009}.
With report-independent tie-breaking, this procedure is strategy-proof for workers, and its outcome is weakly stable under the original weak preferences.
However, the outcome need not be Pareto efficient, since treating tied alternatives as strictly ordered can rule out a Pareto improvement that remains weakly stable under the original preferences.
The following example illustrates this.

Consider three workers $a,b,c$ and two firms $A,B$.
Firm $A$ has two tasks, with values $(4,1)$, $(3,3)$, and $(5,0)$ for workers $a,b,c$, respectively.
It values a set by the best assignment of its workers to distinct tasks, allowing workers to remain unassigned.
Firm $B$ ranks $\{a\}\sim_B\{b\}\succ_B\emptyset$ and finds all other sets unacceptable.
Workers $a$ and $b$ have preference $A\sim B\succ\varnothing$, while $c$ has preference $B\succ A\succ\varnothing$.
If $a$'s tie is broken in favor of $A$ and $b$'s in favor of $B$, worker-proposing deferred acceptance returns $\mu=\{(a,A),(b,B),(c,A)\}$.
Exchanging $a$ and $b$ gives $\nu=\{(a,B),(b,A),(c,A)\}$.
All workers and firm $B$ are indifferent, while $A$'s value rises from $5+1=6$ to $5+3=8$, since $c$ performs task 1 and $b$ replaces $a$ at task 2.
Both matchings are weakly stable because $a$ and $b$ are matched to top-ranked firms and $B$ finds $c$ unacceptable, but $\nu$ Pareto dominates $\mu$.

A natural approach is to Pareto improve the initial matching while preserving weak stability.
Stable improvement cycles can yield matchings that admit no Pareto improvement within the set of weakly stable matchings~\cite{ErdilErgin2008,ErdilErgin2017,BIK2025}.
However, a mechanism based on such improvements need not be strategy-proof~\cite{ErdilErgin2008}.
Thus, resolving ties before selecting a matching may sacrifice Pareto efficiency, while improving the resulting matching may sacrifice strategy-proofness.

We study a many-to-one matching market in which workers have weak preferences over firms and firms have public weak preferences over sets of workers.
We evaluate Pareto efficiency with respect to both workers and firms.
When firm preferences contain ties, no worker strategy-proof mechanism always selects a weakly stable matching that is not Pareto dominated for workers by another weakly stable matching~\cite{ErdilErgin2008}.
This impossibility does not directly apply under our criterion, because a Pareto improvement must also leave every firm weakly better off.
Indeed, weak stability, Pareto efficiency for workers and firms, and worker weak group strategy-proofness are compatible in one-to-one matching and in college admissions when each college's preference consists of the comparisons implied by responsiveness from its ranking of individual students~\cite{DLP2017,DLP2017arXiv}.
Igarashi et al.~\cite{IgarashiEtAl2026} obtain the same guarantees in a many-to-one model in which firms compare sets of workers by the stochastic-dominance relation induced by their rankings of individual workers.
We ask whether these properties remain compatible when firm preferences over sets of workers may be nonresponsive.
For example, firm $A$ above prefers $\{a\}$ to $\{b\}$ but prefers $\{b,c\}$ to $\{a,c\}$.

\subsection{Our results}
We assume that each firm's weak order over sets of workers can be represented by a set function satisfying the \emph{gross substitutes (GS)} condition of Kelso and Crawford~\cite{KelsoCrawford1982}.
GS guarantees the existence of stable outcomes in many-to-one matching both with transfers~\cite{KelsoCrawford1982} and without~\cite{HatfieldMilgrom2005}.
This class includes nonresponsive preferences arising from hiring constraints, such as type-specific quotas, and from diversity policies, such as reserves with overlapping eligibility. It also includes preferences based on position-dependent values, as illustrated by firm $A$'s assignment values above (\Cref{sec:firm-examples}).
Our main result is a deterministic generalized deferred acceptance mechanism that, on this domain, is worker weakly group strategy-proof and selects a Pareto-stable matching, that is, a weakly stable matching that is Pareto efficient for workers and firms (\Cref{thm:generalized-da}).

Generalized deferred acceptance uses firm and worker index orders fixed before workers report and proceeds in rounds.
In each round, every unmatched worker with a remaining acceptable tier moves to her next tier, an indifference class of acceptable firms, while matched workers may change firms within their current tiers.
Among the matchings that respect the current tiers, the firms choose by serial dictatorship in their index order, in which each firm retains all matchings that are best for it among those retained by earlier firms.
The workers then do the same in their index order.
The procedure uses only ordinal comparisons, and the number of rounds is at most the number of worker--firm pairs.
Each round can be implemented in strongly polynomial time given a value oracle for each GS representation (\Cref{app:implementation}), although the output does not depend on the choice of representations.

Under serial dictatorship, a firm chooses among the sets it can obtain in the matchings retained by earlier firms rather than among the subsets of an available set, and the usual arguments for deferred acceptance based on choice functions do not apply.
Our analysis instead extends the tiered-slope market of Domani\c{c}, Lam, and Plaxton~\cite{DLP2017}, a one-to-one market with imperfectly transferable utility, to an auxiliary many-to-one market in which workers are sellers and firms are GS buyers.
We show that equilibrium matchings of this market are Pareto-stable under the original preferences and, using Schlegel's existence and incentive results~\cite{Schlegel2022}, that every deterministic mechanism selecting a seller-optimal one is worker weakly group strategy-proof (\Cref{thm:auxiliary-guarantees}).
Our generalized deferred acceptance returns such a matching at every report profile (\Cref{thm:terminal-equivalence}), which establishes both properties.
The outcome of our mechanism also resists deviations by coalitions of several firms, which subsumes both weak stability and Pareto improvements in which some firm strictly improves (\Cref{prop:coalition-stability}).
The profiles of workers' assigned tiers across equilibrium matchings form a distributive lattice, and our mechanism attains its greatest element (\Cref{cor:rank-lattice}).
Moreover, no strategy-proof mechanism that respects the firms' index order within tiers uniformly Pareto dominates our mechanism for workers (\Cref{thm:undominance}).

Our assumption concerns firms' full preferences over sets of workers, not only their choice correspondences, on which the matching literature typically imposes conditions such as substitutability~\cite{HatfieldMilgrom2005}.
\Cref{app:pi-lad-failure} presents a market in which every firm's choice correspondence is induced by a GS function, yet generalized deferred acceptance returns a matching that is not weakly stable.

\subsection{Related work}

Closest to our work are Domani\c{c}, Lam, and Plaxton~\cite{DLP2017} and Igarashi et al.~\cite{IgarashiEtAl2026}, whose guarantees we recalled above.
The latter also obtain firm-side fairness under quotas and unacceptable assignments.
Our procedure moves workers down their preference tiers as the latter does, and we derive our guarantees from a tiered-slope market inspired by the former.
Our framework therefore unifies the two settings by extending both components to GS firm preferences.
Unlike both earlier models, which derive firm preferences from rankings of individual workers, our framework also covers nonresponsive firm preferences (\Cref{sec:firm-examples}).


Early work on matching with indifferences includes Sotomayor~\cite{Sotomayor1999}.
Erdil and Ergin~\cite{ErdilErgin2008,ErdilErgin2017} study Pareto improvements among weakly stable matchings through cycles and chains.
For bundle comparisons induced by individual rankings, Kamiyama~\cite{Kamiyama2014} develops a polynomial-time algorithm for Pareto-stable matching and extends it to one-sided matroid constraints~\cite{Kamiyama2019}.
Erdil and Kumano~\cite{ErdilKumano2019}, with the corrigendum of Erdil et al.~\cite{ErdilEtAl2022}, consider substitutable choice correspondences with ties.
In recent work, Bando, Imamura, and Kawase~\cite{BIK2025} allow ties on the school side, with choice correspondences satisfying path independence and the law of aggregate demand, and evaluate efficiency for students alone within the set of weakly stable matchings.
We address worker incentives under possibly nonresponsive full preferences over sets of workers, and our efficiency criterion ranges over all matchings.

Substitutability and the law of aggregate demand support individual and weak group strategy-proofness of deferred acceptance when worker preferences are strict~\cite{HatfieldMilgrom2005,HatfieldKojima2009}.
Kojima, Tamura, and Yokoo~\cite{KTY2018} aggregate hospital preferences and distributional constraints into a single GS function, which implies both conditions, and show that generalized deferred acceptance is then strategy-proof for doctors and yields the doctor-optimal stable matching.
Their approach is part of a broad literature in which GS, known in discrete convex analysis as \Mnat-concavity, and its generalizations serve as standard assumptions for establishing stability and efficiency (see, e.g.,~\cite{FujishigeTamura2007,MurotaYokoi2015,IwataYokoi2020,ImamuraKawase2024,ImamuraKawase2025,ImamuraKawase2026}).
Most of these works assume strict preferences, under which a GS function yields a unique choice from every set of available workers.
Imamura and Kawase~\cite{ImamuraKawase2026} allow weak preferences but study stable and fair random allocations.
We keep the GS assumption but allow indifferences on both sides, where choices are no longer unique, and obtain a deterministic mechanism that is Pareto efficient with respect to each firm's own preference.


Our auxiliary market has imperfectly transferable utility, and our guarantees come from its competitive equilibria.
Under unit demands, equilibrium existence and the structure of the equilibrium set have long been known~\cite{Kaneko1982,Quinzii1984,DemangeGale1985}, and an equilibrium can be computed~\cite{Alkan1989,Alkan1992}.
More recently, Fleiner et al.~\cite{FleinerEtAl2019} establish existence for combinatorial markets, and Schlegel~\cite{Schlegel2022} analyzes the structure of these equilibria, including the side-optimal existence and incentive results that our proof uses.
Balkanski, En, and Faenza~\cite{BalkanskiEtAl2024} and Oki and Sato~\cite{OkiSato2026} propose algorithms for computing equilibria in such markets, whereas our procedure exploits the tiered structure to reach an equilibrium matching faster, without computing transfers.


\section{Model}\label{sec:matching-model}

\subsection{Ordinal matching market}

We study a finite many-to-one ordinal matching market $(I,J,(\succsim_i)_{i\in I},(\succsim_j)_{j\in J})$ in which preferences may contain ties.
Here, $I$ and $J$ are the sets of workers and firms, respectively.
For each worker $i$, $\succsim_i$ is a weak order (a complete and transitive relation) on $J\cup\{\varnothing\}$, where $\varnothing$ denotes remaining unmatched.
Write $\succ_i$ and $\sim_i$ for the strict-preference and preference-tie parts of $\succsim_i$, respectively.
Thus, $j\sim_i k$ means that worker $i$ is indifferent between $j$ and $k$.
For $S\subseteq I$, write $\succsim_S=(\succsim_i)_{i\in S}$ for the preference profile of workers in $S$.
Workers report their preferences to the mechanism.
For each firm $j$, $\succsim_j$ is a weak order on $2^I$ and represents its preference over sets of workers.\footnote{\label{fn:preorders}Preferences may instead be preorders (reflexive and transitive, but not necessarily complete).
For workers, require comparability with $\varnothing$ to preserve acceptability, and complete each report using a fixed option order while preserving every weak and strict comparison.
For firms, assume analogous GS-representable extensions exist and fix them publicly before reports.
Pareto stability and worker weak group strategy-proofness, defined below, then hold with respect to the original preorders.}
The firm preference profile $(\succsim_j)_{j\in J}$ is public information and firms do not report it.
Write $\succ_j$ and $\sim_j$ for the strict-preference and preference-tie parts of $\succsim_j$, respectively.
Thus, $X\sim_jY$ means that firm $j$ is indifferent between $X$ and $Y$.
For a set $A\subseteq I$ of available workers, firm $j$'s choice correspondence is $\cC_j(A)=\{X\subseteq A\mid\text{there is no }Y\subseteq A\text{ such that }Y\succ_jX\}$.

A matching is a set $\mu\subseteq I\times J$ such that each worker appears in at most one pair.
With a slight abuse of notation, write $\mu(i)=j$ if $(i,j)\in\mu$ and $\mu(i)=\varnothing$ if worker $i$ appears in no pair.
For each firm $j$, write $\mu(j)=\{i\in I\mid(i,j)\in\mu\}$.
A firm may employ several workers.
A mechanism $\varphi$ assigns a matching to every profile of reported worker preferences.
For a fixed profile, write $\mu=\varphi(\succsim_I)$ and $\varphi_i(\succsim_I)=\mu(i)$.
A firm $j$ and a set $T\subseteq I$ of workers form a \emph{blocking coalition} for a matching $\mu$ if $T\succ_j\mu(j)$ and every worker $i\in T\setminus\mu(j)$ strictly prefers $j$ to $\mu(i)$.
A matching $\mu$ is
\begin{itemize}[itemsep=0.3em, topsep=.6em]
 \item \emph{individually rational} if $\mu(i)\succsim_i\varnothing$ for every worker $i$ and $\mu(j)\in\cC_j(\mu(j))$ for every firm $j$, the latter meaning that no firm strictly prefers a subset of its assigned workers,
 \item \emph{weakly stable} if it is individually rational and has no blocking coalition,
 \item \emph{Pareto efficient} if no matching $\nu$ satisfies $\nu(i)\succsim_i\mu(i)$ for every worker $i$ and $\nu(j)\succsim_j\mu(j)$ for every firm $j$ with at least one comparison strict, and
 \item \emph{Pareto-stable} if it is weakly stable and Pareto efficient.
\end{itemize}


A mechanism $\varphi$ is weakly stable or Pareto efficient if its outcome has the corresponding property at every worker preference profile.
It is worker \emph{weakly group strategy-proof (WGSP)} if there do not exist a true worker preference profile $\succsim_I$, a nonempty coalition $S\subseteq I$, and alternative reports $\succsim'_S$ such that $\varphi_i(\succsim'_S,\succsim_{I\setminus S})\succ_i\varphi_i(\succsim_I)$ for every $i\in S$.
Applied to single-worker coalitions, worker WGSP implies that no worker gains by misreporting alone, that is, worker strategy-proofness.%
\footnote{Even in one-to-one markets with strict preferences on both sides, Pareto stability and worker group strategy-proofness are incompatible on the unrestricted domain~\cite{AlcaldeBarbera1994,Ergin2002}.
A mechanism is \emph{worker group strategy-proof} if there do not exist $\succsim_I$, $S\subseteq I$, and $\succsim'_S$ such that $\varphi_i(\succsim'_S,\succsim_{I\setminus S})\succsim_i\varphi_i(\succsim_I)$ for every $i\in S$ and at least one of these comparisons is strict.}

\subsection{Gross substitutes firm preferences}\label{sec:gs-preferences}

Here, we define the GS condition on set functions, which we impose on firm preferences and also use for the buyers' valuations in the auxiliary market of \Cref{sec:transfer-market}.
Let $E$ be a finite ground set.
Consider a set function $u\colon 2^E\to\mathbb{R}\cup\{-\infty\}$ with nonempty effective domain $\dom(u)=\{X\subseteq E\mid u(X)>-\infty\}$.
For a price vector $\bm{p}\in\mathbb{R}^E$, write $p(X)=\sum_{e\in X}p_e$ and define $\cD_u(\bm{p})=\argmax\{u(X)-p(X)\mid X\subseteq E\}$.
The function $u$ satisfies \emph{gross substitutes} (GS) if, for all $\bm{p},\bm{q}\in\mathbb{R}^E$ with $\bm{p}\le\bm{q}$ and every $X\in\cD_u(\bm{p})$, there is $Y\in\cD_u(\bm{q})$ such that $\{e\in X\mid p_e=q_e\}\subseteq Y$. 
This condition is equivalent to the following exchange axiom~\cite{FujishigeYang2003}, \cite[Theorem~3.7]{Murota2016}: 
for all $X,Y\in\dom(u)$ and $e\in X\setminus Y$, there exists $f\in(Y\setminus X)\cup\{\varnothing\}$ such that
\begin{align}
 u(X)+u(Y)\le u(X-e+f)+u(Y+e-f).
\label{eq:exc}
\end{align}
Here and below, when $f=\varnothing$, we interpret $X-e+f$ as $X\setminus\{e\}$ and $Y+e-f$ as $Y\cup\{e\}$.
We call a set function satisfying this exchange axiom a \emph{GS function}.\footnote{These are known in discrete convex analysis as \Mnat-concave functions~\cite{Murota2003,Murota2016}.}

A nonempty family $\cA\subseteq 2^E$ is called a \emph{generalized matroid}, or a \emph{g-matroid}, if, for all $X,Y\in\cA$ and $e\in X\setminus Y$, there exists $f\in(Y\setminus X)\cup\{\varnothing\}$ such that $X-e+f,Y+e-f\in\cA$.
The exchange axiom~\eqref{eq:exc} implies that the effective domain of every GS function $u$ and each demand family $\cD_u(\bm{p})$ are g-matroids. A g-matroid containing $\emptyset$ is the family of \emph{independent sets} of a \emph{matroid}~\cite[Section~3.1]{Murota2016}.
An additive function is minimized over a g-matroid $\cA$ at $X\in\cA$ if and only if no single-element addition, deletion, or exchange within $\cA$ decreases its value~\cite[Theorem~6.26]{Murota2003}.

An integer-valued GS function $v_j\colon 2^I\to\mathbb{Z}\cup\{-\infty\}$ with $v_j(\emptyset)=0$ is a \emph{GS representation} of $\succsim_j$ if, for every $X,Y\subseteq I$, we have $X\succsim_jY$ if and only if $v_j(X)\ge v_j(Y)$.
We say that $\succsim_j$ is \emph{GS-representable} if it admits a GS representation.\footnote{Integer values lose no generality.
Suppose that $\succsim_j$ is represented by a real-valued GS function.
For each instance of the exchange axiom, fix a choice of $f$ for which it holds.
These inequalities, together with equalities for indifferences and gaps of at least one for strict preferences, form a finite system of linear inequalities with integer coefficients on the fixed effective domain.
A scaled copy of the real-valued representation satisfies the system, which therefore has a rational solution, and scaling it gives an integer-valued GS representation.}

\subsection{Examples of GS-representable preferences}\label{sec:firm-examples}

GS-representable preferences can reflect individual rankings, hiring constraints, diversity policies, and position-dependent values.

\paragraph{Additive scores with quotas.}
Additive scores extend individual rankings to preferences over sets of workers.
For a firm $j$ with quota $q_j$, choose integer scores $w_{ij}$ representing its weak ranking of workers, with the outside option assigned score zero.
The resulting function $v_j(X)=\sum_{i\in X}w_{ij}$ for
$|X|\le q_j$, with $v_j(X)=-\infty$ otherwise, is GS.
On sets within the quota, its induced weak order extends the preorder generated by responsiveness, preserving both weak and strict comparisons~\cite{DLP2017,DLP2017arXiv}.
Our guarantees in \Cref{thm:generalized-da} therefore carry over to that preorder, as explained in footnote~\ref{fn:preorders}.

\paragraph{Additive scores with matroid constraints.}
Additive scores also accommodate more general hiring constraints when the feasible sets form a matroid.
If $\cI_j$ is its family of independent sets, the same additive formula, with value $-\infty$ outside $\cI_j$, is GS~\cite[Chapter~6]{Murota2003}.
For example, if every worker has one type, upper quotas for the individual types together with an overall quota define a matroid constraint~\cite[Chapter~4]{Murota2003}.

\paragraph{Laminar concave values.}
Laminar concave values capture preferences over the numbers of workers hired from disjoint or nested groups.
Let $\mathcal{L}$ be a laminar family of subsets of $I$, meaning that any two members are disjoint or one contains the other.
For each $L\in\mathcal{L}$, let $\varphi_L\colon\{0,\ldots,|L|\}\to\mathbb{Z}\cup\{-\infty\}$ be concave with $\varphi_L(0)=0$.
Then $v_j(X)=\sum_{L\in\mathcal{L}}\varphi_L(|X\cap L|)$ is GS~\cite[Note~6.11]{Murota2003}.
For example, if $\varphi_L$ is steeply increasing up to some $t_L$ and constant thereafter, the firm gains from each additional worker from $L$ until it has hired $t_L$ of them, and further workers from $L$ bring no additional value~\cite{KTY2018,BIK2025}.
If $\varphi_L$ instead has a unique maximum at $t_L$, its value increases as the number of hires approaches $t_L$ from either side, reaching its maximum at $t_L$~\cite{BIK2025}.
Taking $\mathcal{L}$ to consist of singletons recovers additive scores.

%

\paragraph{Assignment values.}
Assignment values allow a worker's contribution to depend on the position she fills~\cite{HatfieldMilgrom2005}.
For $X\subseteq I$, define $v_j(X)$ as the maximum total value of an assignment of every worker in $X$ to a distinct eligible position, and $-\infty$ if no such assignment exists.
With integer values for worker--position pairs, this function is GS~\cite[Section~3.6]{Murota2016}.
Adding enough zero-value positions open to all workers allows them to remain unassigned to the original positions, as in~\Cref{sec:introduction}.

Such valuations can also represent policies that reserve some positions for specified worker types while leaving others open to all workers.
Although a worker may be eligible for several reserves, she can fill only one position and therefore counts toward at most one reserve.
Without the zero-value positions above, the number of positions acts as an overall quota, and giving reserved positions sufficiently large bonuses prioritizes filling them within this quota~\cite{BIK2025}.
This treatment of overlapping types arises in affirmative action in India~\cite{SonmezYenmez2022} and Brazil~\cite{AygunBo2021}.

\section{Generalized Deferred Acceptance}\label{sec:generalized-da}

In this section, we construct generalized deferred acceptance using serial dictatorship over feasible matchings.
The procedure lets unmatched workers move down their preference tiers while allowing matched workers to change firms within their current tier.

Fix a firm index order $J=\{j_1,\ldots,j_m\}$ and a worker index order $I=\{i_1,\ldots,i_n\}$ before workers report.
A firm $j$ is acceptable to worker $i$ if $j\succsim_i\varnothing$.
Her \emph{current eligible tier} $F_i$ is either $\emptyset$ or an indifference class of acceptable firms.
Let $\operatorname{next}_i(F_i)$ be worker $i$'s most-preferred acceptable tier when $F_i=\emptyset$ and her most-preferred acceptable tier strictly below $F_i$ otherwise, undefined if no such tier exists.
When worker $i$ advances, $\operatorname{next}_i(F_i)$ replaces $F_i$, and earlier tiers become unavailable.

For a profile $F=(F_i)_{i\in I}$, let $\cM(F)$ be the set of matchings $\nu$ satisfying $\nu(i)\in F_i\cup\{\varnothing\}$ for every worker $i$.
Set $\cM_0(F)=\cM(F)$.
Firms and then workers successively retain their best remaining matchings.
For $t=1,\ldots,m$, firm $j_t$ retains
\begin{align*}
 \cM_t(F)=\{\nu\in\cM_{t-1}(F)\mid
 \nu(j_t)\succsim_{j_t}\xi(j_t)\text{ for every }\xi\in\cM_{t-1}(F)\}.
\end{align*}
This is \emph{serial dictatorship over feasible matchings}, retaining all matchings that are lexicographically maximal in the firm index order.
For $q=1,\ldots,n$, worker $i_q$ then retains
\begin{align*}
 \cM_{m+q}(F)=\{\nu\in\cM_{m+q-1}(F)\mid
 \nu(i_q)\succsim_{i_q}\xi(i_q)\text{ for every }\xi\in\cM_{m+q-1}(F)\}.
\end{align*}
Since a worker is indifferent among the firms in her tier, this choice keeps only the matchings in which she is matched, if any exist, when the tier is strictly above $\varnothing$, and leaves the candidates unchanged otherwise.
Select $\Psi(F)\in\cM_{m+n}(F)$ by lexicographically minimizing $(\nu(i_1),\ldots,\nu(i_n))$, using the order $j_1,j_2,\ldots,j_m,\varnothing$ to compare entries.

The generalized deferred acceptance procedure is as follows.
\begin{procbox}[title=Generalized Deferred Acceptance,
  fonttitle=\bfseries, coltitle=black, colbacktitle=white]
\begin{enumerate}[label=\textbf{Step~\arabic*.},leftmargin=*,topsep=0pt,itemsep=0.35em,parsep=0pt]
 \item Initialize $F_i=\emptyset$ and $\mu(i)=\varnothing$ for every worker $i$.
 \item If $\operatorname{next}_i(F_i)$ is undefined for every unmatched worker $i$, return $\mu$.
 \item For every unmatched worker $i$ for whom $\operatorname{next}_i(F_i)$ is defined, replace $F_i$ by $\operatorname{next}_i(F_i)$.
 Set $\mu=\Psi(F)$ and return to Step~2.
\end{enumerate}
\end{procbox}
Each execution of Step~3, called a round, moves at least one worker to her next acceptable tier.
Since no worker revisits a tier, the procedure stops after at most $|I|\cdot|J|$ rounds.

Whenever $\Psi(F)$ is evaluated, every worker with an acceptable firm has a nonempty tier.
If a worker has a next acceptable tier, her current tier lies strictly above $\varnothing$, and the workers' choices therefore fix whether she is matched throughout $\cM_{m+n}(F)$.
Hence, the workers who advance, and thus all visited tier profiles, do not depend on which matching in $\cM_{m+n}(F)$ is selected.

Under GS-representable firm preferences, if worker preferences are strict and each firm has a unique best subset of every set of available workers, the procedure coincides with worker-proposing deferred acceptance, which here is equivalent to the cumulative offer process of Hatfield and Milgrom~\cite{HatfieldMilgrom2005}.
Each unmatched worker proposes to her next acceptable firm, and each firm chooses from its previously retained workers and new applicants. 
Since each worker is eligible for only one firm, serial dictatorship reduces to the firms' separate choices, which are unique by assumption. 
In the one-to-one setting, allowing indifferences on both sides, the procedure coincides with the algorithm of Igarashi et al.~\cite{IgarashiEtAl2026}.\footnote{In many-to-one markets, firms in the procedure of Igarashi et al.~\cite{IgarashiEtAl2026} choose in a round-robin order rather than by serial dictatorship.
Our analysis also covers this choice once the firms' separate GS representations are replaced by a single GS function for all firms, under the conditions given in \Cref{sec:discussion}.}

\Needspace{6\baselineskip}
Under the GS assumption, we obtain the following guarantees.

\begin{theorem}\label{thm:generalized-da}
Suppose every firm preference is GS-representable.
For any firm and worker index orders fixed before workers report, generalized deferred acceptance is worker WGSP and selects a Pareto-stable matching.
Given a value oracle for a GS representation of each firm preference, it runs in strongly polynomial time.
\end{theorem}

\Cref{fig:roadmap} sketches the proof.
We introduce an auxiliary market in which the workers are sellers and the firms are buyers.
A firm pays a transfer to each worker it hires, but utility is only imperfectly transferable.
The rate at which a worker converts transfers into utility depends on her partner, and this dependence encodes her ordinal preferences. 
In \Cref{sec:equilibrium-analysis}, we show that every deterministic mechanism selecting a seller-optimal equilibrium matching is Pareto-stable and worker WGSP.
In \Cref{sec:gda-analysis}, we show that generalized deferred acceptance returns one of these matchings.
Under imperfectly transferable utility, computing such a matching generally requires sophisticated algorithms~\cite{BalkanskiEtAl2024,OkiSato2026}, but the tiered structure lets generalized deferred acceptance reach it in purely ordinal terms, without computing transfers or payoffs.


\begin{figure}[!ht]
\centering
\begin{tikzpicture}[
  font=\small,
  box/.style={draw, fill=blue!2, rounded corners=2pt, align=center, inner sep=5pt, text width=70mm},
  lab/.style={font=\scriptsize, align=left, inner sep=3pt, text width=62mm,black!80},
  to/.style={->, >={Stealth[length=5pt]}, thick},
]

\node[box] (gda)
  {\textbf{Generalized deferred acceptance}\\[1pt]
   \footnotesize \Cref{sec:generalized-da}};

\node[box, below=15mm of gda] (aux)
  {\textbf{Seller-optimal equilibrium matchings}\\[1pt]
   \footnotesize of the auxiliary market,
   \Cref{sec:equilibrium-analysis}};

\node[box, below=15mm of aux] (prop)
  {\textbf{Pareto stability and worker WGSP}\\[1pt]
   \footnotesize under the original preferences};

\draw[to] ([xshift=12mm]gda.south west) --
  node[lab, right=4pt]
  {\textbf{\Cref{thm:terminal-equivalence}}\\[1mm]
   matchings retained at termination are\\
   the seller-optimal equilibrium matchings}
  ([xshift=12mm]aux.north west);

\draw[to] ([xshift=12mm]aux.south west) --
  node[lab, right=4pt]
  {\textbf{\Cref{thm:auxiliary-guarantees}}\\[1mm]
   deterministic seller-optimal mechanisms are \\\mbox{Pareto-stable}
   and worker WGSP}
  ([xshift=12mm]prop.north west);

\draw[to, dashed] (gda.west)--([xshift=-5mm]gda.west)
  to
  node[lab, left=2pt, text width=22mm, align=right]
  {\small\textbf{\Cref{thm:generalized-da}}}
  ([xshift=-5mm]prop.west)--(prop.west);

\end{tikzpicture}
\caption{Proof sketch of \Cref{thm:generalized-da}, which follows from Theorems~\ref{thm:auxiliary-guarantees} and~\ref{thm:terminal-equivalence}. 
}
\label{fig:roadmap}
\end{figure}
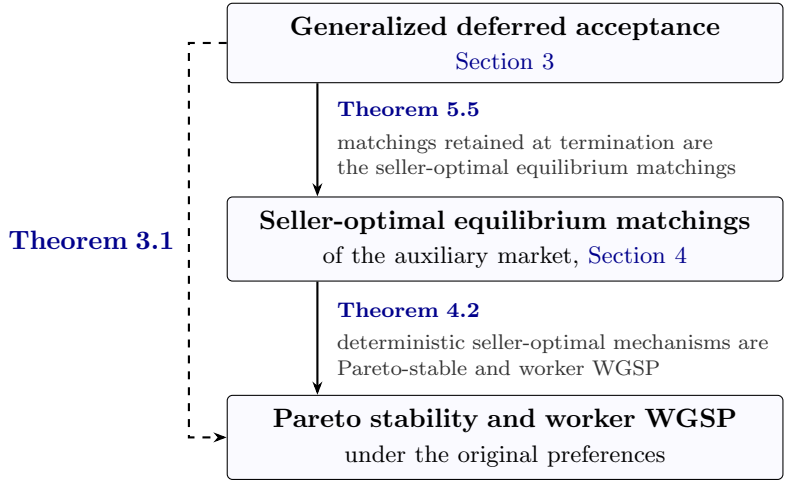

\begin{example}\label{ex:gda-execution}
We illustrate the procedure using the three-worker, two-firm market from \Cref{sec:introduction}, with the firm order $A,B$.

In the first round, the eligible tiers are $F_a=F_b=\{A,B\}$ and $F_c=\{B\}$.
Among the available sets, firm $A$ most prefers $\{a,b\}$, while $B$ prefers remaining empty to hiring $c$.
The selected matching is $\{(a,A),(b,A)\}$, leaving $c$ unmatched.

In the second round, $c$ advances to $F_c=\{A\}$, while $F_a=F_b=\{A,B\}$ remain unchanged.
Firm $A$ is indifferent between $\{b,c\}$ and $\{a,b,c\}$ and strictly prefers both to every other available set.
After $A$'s choice, three matchings remain: $\{(a,B),(b,A),(c,A)\}$, $\{(b,A),(c,A)\}$, and $\{(a,A),(b,A),(c,A)\}$.
Firm $B$ prefers hiring $a$ to remaining empty and therefore retains only the first matching.
All workers are matched, and the procedure stops after this round at $\nu$ from the Introduction.
Worker $a$ moves from $A$ to $B$ within her unchanged tier.
\end{example}

\section{Equilibria of the Auxiliary Market}
\label{sec:equilibrium-analysis}

In this section, we construct the auxiliary market outlined in \Cref{sec:generalized-da} and study its equilibria.
In this market, the workers are sellers, the firms are buyers, and utility is only imperfectly transferable.
The market is parameterized by the workers' outside payoffs, that is, their utilities from remaining unmatched.
For a baseline choice of these payoffs (\Cref{sec:worker-tiers}), we show that every equilibrium matching is Pareto-stable under the original preferences and that every deterministic mechanism selecting a seller-optimal equilibrium matching is worker WGSP (\Cref{thm:auxiliary-guarantees}).
\Cref{sec:gda-analysis} varies the outside payoffs to track the rounds of generalized deferred acceptance.

\subsection{Markets with transfers}\label{sec:transfer-market}

Consider a market with transfers $(I,J,E,(U_i)_{i\in I},(V_j)_{j\in J})$, where $I$ and $J$ are finite disjoint sets of sellers and buyers and $E$ is a finite set of trades, each between one seller and one buyer.
Write $\iota(x)$ for the seller of a trade $x$, and let $E_i$ and $E_j$ collect the trades of seller $i$ and buyer $j$.
Each seller supplies at most one trade, while a buyer may take several.
A transfer vector $\bm{t}\in\mathbb{R}^{E}$ specifies the payment $t_x$ from the buyer to the seller on every possible trade $x$.

Seller $i$ has a utility function $U_i\colon(E_i\times\mathbb{R})\cup\{\varnothing\}\to\mathbb{R}$.
The value $U_i(x,t_x)$ is her utility from trade $x$ and transfer $t_x$, and $U_i(\varnothing)$ is her utility from remaining unmatched.
For every $x\in E_i$, the function $U_i(x,\cdot)$ is continuous and strictly increasing.
Buyer $j$ has a quasilinear utility function $V_j\colon 2^{E_j}\times\mathbb{R}^{E_j}\to\mathbb{R}\cup\{-\infty\}$ given by $V_j(X,\bm{t}_j)=\tilde{v}_j(X)-\sum_{x\in X}t_x$, where $\bm{t}_j=(t_x)_{x\in E_j}$ and $\tilde{v}_j\colon 2^{E_j}\to\mathbb{R}\cup\{-\infty\}$ is a GS function with $\tilde{v}_j(\emptyset)>-\infty$.
Here, $\tilde{v}_j(X)$ is the buyer's valuation of the trades in $X$, and $\sum_{x\in X}t_x$ is its total payment.
Buyer $j$'s demand correspondence is 
\begin{align*} 
\cD_j(\bm{t})=\argmax\{V_j(X,\bm{t}_j)\mid X\subseteq E_j\}. 
\end{align*}
Since $\tilde{v}_j(\emptyset)>-\infty$, the maximum is finite and no set of value $-\infty$ is~demanded.
The buyers' utilities are fixed and public, while the sellers report their utility functions.

An \emph{allocation} $X\subseteq E$ contains at most one trade per seller, and we write $X_j\coloneqq X\cap E_j$.
A pair $(X,\bm{t})$ of an allocation and a transfer vector is a \emph{competitive equilibrium} if, given $\bm{t}$, each seller's assignment maximizes her utility over her trades and the outside option, and $X_j\in\cD_j(\bm{t})$ for each buyer $j$.
It is \emph{seller-optimal} if it maximizes every seller's utility simultaneously among all competitive equilibria.
An allocation is an \emph{equilibrium allocation} if it forms a competitive equilibrium with some transfer vector, and a \emph{seller-optimal equilibrium allocation} if it forms a seller-optimal one.
Seller-optimal equilibria give each seller the same utility but need not share an allocation, and a mechanism is \emph{seller-optimal} if it selects a seller-optimal equilibrium at every report profile.

The analysis below builds on the existence and incentive results of Schlegel~\cite{Schlegel2022}, which require bounded compensating variations.
We therefore assume $\inf\{t\in\mathbb{R}\mathrel{:} U_i(x,t)>U_i(\varnothing)\}>-\infty$ for every $i\in I$ and $x\in E_i$, and sellers report such utility functions.
The buyers satisfy the corresponding condition since $\emptyset$ is feasible and each $\tilde{v}_j$ takes finitely many finite values.
These results also assume full substitutability and the laws of aggregate demand and supply.
Each buyer's quasilinear utility with a GS valuation satisfies these conditions, and so does each seller's utility function, since a seller supplies at most one trade.
Schlegel~\cite{Schlegel2022} states the incentive result for unit-demand terminal buyers and notes that it holds symmetrically for unit-supply terminal sellers, which is the case here.
His framework also allows utility functions to take the value $-\infty$, as the buyers' valuations do here.
A mechanism is \emph{seller WGSP} if no coalition of sellers can misreport so that every member is strictly better off.
We use these results in the following form.
\begin{theorem}[{Schlegel~\cite[Theorems~2 and~3 and Remarks~4 and~5]{Schlegel2022}}]\label[theorem]{thm:schlegel}
In the market with transfers $(I,J,E,(U_i)_{i\in I},(V_j)_{j\in J})$, a seller-optimal equilibrium exists.
Furthermore, every deterministic seller-optimal mechanism is seller WGSP.
\end{theorem}

\subsection{A tiered auxiliary market}\label{sec:worker-tiers}
Here, we construct a market with transfers from the ordinal matching market $(I,J,(\succsim_i)_{i\in I},(\succsim_j)_{j\in J})$ of \Cref{sec:matching-model}.
This \emph{auxiliary market} is $(I,J,E,(U_i)_{i\in I},(V_j)_{j\in J})$ with $E=I\times J$, where the sellers are the workers, the buyers are the firms, each pair $(i,j)$ is a trade with $\iota(i,j)=i$, and an allocation is a matching.
For $X\subseteq E$, write $X(j)=\{i\in I\mathrel{:}(i,j)\in X\}$, hence $X_j=X\cap E_j=X(j)\times\{j\}$.
Before workers report, fix a GS representation $v_j$ of each firm preference and a positive integer weight $\alpha_j$, and set $v(X)=\sum_{j\in J}\alpha_jv_j(X(j))$ for $X\subseteq E$.
The function $v$ is GS because its summands have disjoint ground sets.
Its effective domain consists of the sets $X$ for which every $v_j(X(j))$ is finite.
\Cref{thm:auxiliary-guarantees} below holds for any such weights, and \Cref{sec:gda-analysis} chooses them so that maximizing $v$ reproduces the firms' serial dictatorship.

We construct the buyers' valuations to reproduce the selection in each round of generalized deferred acceptance. They rank matchings first by \(v\), and then lexicographically by matched workers in their index order.
The valuations do not encode the final selection of $\Psi(F)$, which only picks one of the retained matchings, and \Cref{thm:auxiliary-guarantees} below applies to every deterministic selection.
For the worker index order $I=\{i_1,\dots,i_n\}$, fix positive integer weights $\omega_i$ with $\omega_{i_q}-1>\sum_{r>q}\omega_{i_r}$ for every $q$, and set $\bar{\omega}=\sum_{i\in I}\omega_i$.
For $X\subseteq E$, write $\omega(X)=\sum_{x\in X}\omega_{\iota(x)}$, counting each trade separately.
Fix an integer $\kappa>\max\{m\bar{\omega},\,\bar{\omega}+n\}$.\footnote{For instance, $\omega_{i_q}=2^{n-q+1}$ encodes the order, since $\sum_{r>q}\omega_{i_r}=\omega_{i_q}-2$, and $\kappa=(m+1)\bar{\omega}+1$ suffices.}
For each $j\in J$ and $X\subseteq E_j$, define the valuation
\begin{align*}
 \tilde{v}_j(X)=\kappa\alpha_jv_j(X(j))+\omega(X).
\end{align*}
Each $\tilde{v}_j$ is GS, being a positive scaling of $v_j$ plus an additive function.
Hence, so is their sum $\tilde{v}(X)=\sum_{j\in J}\tilde{v}_j(X_j)=\kappa v(X)+\omega(X)$.
Since $v$ is integer-valued on its domain and $0\le\omega(X)\le m\bar{\omega}<\kappa$, maximizing $\tilde{v}$ first maximizes $v$ and uses $\omega$ only to break ties.
Set $\bar{v}
\coloneqq\kappa\max\left\{1,\max_{X\in\dom(v)}v(X)\right\}
+m\bar{\omega}$ 
as an upper bound on $\tilde{v}$.
These constants play separate roles.
The weights $\alpha_j$ order the firms inside $v$, the weights $\omega_i$ break the remaining ties among workers in their index order, and $\kappa$ makes one unit of $v$ outweigh all $\omega$ terms and transfers below one.
The constant $\lambda$ defined next makes a one-tier difference in a worker's ranking outweigh every transfer up to $\bar{v}$.

A worker's utility from a transfer depends on which firm pays it, and the higher she ranks that firm, the more a given transfer yields.
For each reported preference $\succsim_i$ and $j\in J\cup\{\varnothing\}$, define the rank $r_i(j)=\lvert\{k\in J\cup\{\varnothing\}\mathrel{:} j\succ_i k\}\rvert$, hence $j\succsim_i k$ if and only if $r_i(j)\ge r_i(k)$.
Write $r_i(x)=r_i(j)$ for $x=(i,j)$.
Setting $\lambda\coloneqq2\bar{v}+1$, define each worker $i$'s utility from a trade $x\in E_i$ by
\begin{align*}
 U_i(x,t_x)=t_x\lambda^{r_i(x)}.
\end{align*}
A firm in a higher tier therefore delivers her the same positive payoff with a transfer smaller by a factor of at least $\lambda$.
Set $U_i(\varnothing)=\beta_i$ for an outside payoff vector $\bm{\beta}>\bm{0}$.
The \emph{baseline} market takes $\bm{\beta}=\bm{\beta}^0$ with $\beta_i^0\coloneqq\omega_i\lambda^{r_i(\varnothing)}$.
In \Cref{sec:gda-analysis}, only $\bm{\beta}$ varies across the procedure's~rounds. 
For every $\bm{\beta}>\bm{0}$, the seller utilities are continuous, strictly increasing in transfers, and satisfy bounded compensating variations, hence the auxiliary market falls within the setting of \Cref{thm:schlegel}.

Since all representations, weights, and constants are fixed before workers report, the auxiliary market depends on the reports only through each worker's own trade utilities and baseline outside payoff.
We now state the main result of this section, deferring its proof to \Cref{sec:ordinal-guarantees}.
\begin{theorem}\label[theorem]{thm:auxiliary-guarantees}
Every equilibrium matching of the baseline market is Pareto-stable under the original worker and firm preferences.
Furthermore, every deterministic mechanism selecting a seller-optimal equilibrium matching of this market at each report profile is worker WGSP. 
\end{theorem}

Our construction is inspired by Domani\c{c}, Lam, and Plaxton~\cite{DLP2017}, who build a tiered-slope market within the generalized assignment game of Demange and Gale~\cite{DemangeGale1985}.
Their market is one-to-one, and either side may carry the tiered slopes.
In our market, a firm hires several workers and must therefore be a buyer, while each worker supplies at most one trade and is a seller.
The many-to-one structure thus removes the choice and leaves the tiered slopes on the worker side.

\subsection{Equilibrium payoff bounds}\label{sec:tier-bounds}

Here, we show that a worker's payoff at a seller-optimal equilibrium of the baseline market lies in $[\lambda^{r},\lambda^{r+1})$, where $r$ is the rank of the firm she is matched to.
Her payoff therefore identifies her tier, and \Cref{sec:ordinal-guarantees} restates the equilibrium conditions in terms of the original preferences.

We first express equilibria in terms of worker payoffs rather than a transfer for each trade.
Payoffs ignore transfers on unassigned trades, which affect no agent's utility, and can be compared coordinatewise across equilibria.
Since $U_i(x,t_x)=t_x\lambda^{r_i(x)}$, worker $i$ obtains payoff $p_i$ from her trade with firm $j$ exactly when $t_{ij}=\lambda^{-r_i(j)}p_i$.
For $\bm{p}\in\mathbb{R}_+^I$, we use these transfers $\bm{t}$ throughout and define the joint demand at $\bm{t}$ by
\begin{align*}
 \hat{\cD}(\bm{p})=\argmax\left\{\tilde{v}(X)-\sum_{(i,j)\in X}\lambda^{-r_i(j)}p_i\mathrel{:} X\subseteq E\right\}.
\end{align*}
Because the buyers' trade sets $E_j$ are disjoint, $X\in\hat{\cD}(\bm{p})$ if and only if $X_j\in\cD_j(\bm{t})$ for every buyer $j$.
At these transfers, worker $i$ obtains the same payoff $p_i$ from all her trades.
Each worker's assignment is therefore optimal if and only if $p_i\ge\beta_i=U_i(\varnothing)$, with equality for unmatched workers.

A matching $X$ supports an equilibrium payoff vector $\bm{p}$ if and only if
\begin{align}
 X\in\hat{\cD}(\bm{p}),\qquad
 \bm{p}\ge\bm{\beta},\qquad\text{and}\qquad
 p_i=\beta_i\ \text{for every $i$ with $X(i)=\varnothing$}.
 \label{eq:standard-equilibrium}
\end{align}
Indeed, in an equilibrium $(X,\bm{t}')$ with payoffs $\bm{p}$, worker optimality implies $\bm{t}'\le\bm{t}$ with equality on assigned trades, and raising $\bm{t}'$ to $\bm{t}$ keeps each buyer's assigned set demanded.
Conversely, $(X,\bm{t})$ is an equilibrium whenever $(X,\bm{p})$ satisfies \eqref{eq:standard-equilibrium}.
Henceforth, we write an equilibrium as $(X,\bm{p})$ and $x_i$ for the trade assigned to a matched worker $i$.

We begin with upper bounds, which hold at every equilibrium.
\begin{lemma}\label[lemma]{lem:payoff-upper}
Let $(X,\bm{p})$ be an equilibrium with outside payoffs $\bm{\beta}>\bm{0}$.
Then, $0<t_x\le\bar{v}$ for every trade $x$ of every $Y\in\hat{\cD}(\bm{p})$.
In particular, $0<p_i\le\bar{v}\lambda^{r_i(x_i)}$ for every matched worker $i$.
\end{lemma}
\begin{proof}
Each $t_x$ is positive because $\bm{p}\ge\bm{\beta}>\bm{0}$.
For $Y\in\hat{\cD}(\bm{p})$, comparing $Y$ with $\emptyset$ in the definition of $\hat{\cD}(\bm{p})$ shows $\sum_{x\in Y}t_x\le\tilde{v}(Y)\le\bar{v}$, and each term is at most this sum.
Taking $Y=X$ yields the second bound, since $p_i=t_{x_i}\lambda^{r_i(x_i)}$.
\end{proof}

Since $\bar{v}<\lambda$, we have $p_i<\lambda^{r_i(x_i)+1}$ for every matched worker $i$ by \Cref{lem:payoff-upper}.
By contrast, the lower bound $p_i\ge\lambda^{r_i(x_i)}$ can fail at some equilibria, because a matched worker may be paid only her outside payoff $\beta_i$.
We show that it holds whenever no equilibrium with the same matching has coordinatewise larger payoffs, in particular at the seller-optimal equilibria.
Otherwise, the payoffs of a suitable set of workers can be raised, each within her assigned tier, without changing the equilibrium matching.
For an equilibrium $(X,\bm{p})$, a set $R$ of matched workers, and $\varepsilon>0$, raising the transfer on $x_i$ by $\varepsilon$ for every $i\in R$ produces the payoff vector $\bm{p}^{\varepsilon}$ with $p_i^{\varepsilon}=p_i+\varepsilon\lambda^{r_i(x_i)}$ for $i\in R$ and $p_i^{\varepsilon}=p_i$ otherwise.
The next lemma provides conditions under which $X$ remains an equilibrium matching at $\bm{p}^{\varepsilon}$.

\begin{lemma}\label[lemma]{lem:raising}
Let $(X,\bm{p})$ be an equilibrium with outside payoffs $\bm{\beta}>\bm{0}$, and let $R$ be a nonempty set of matched workers such that
\begin{enumerate}
 \item[(i)] $X-x_i\notin\hat{\cD}(\bm{p})$ for every $i\in R$, and
 \item[(ii)] for every $i\in R$ and $y\in E\setminus X$ with $X-x_i+y\in\hat{\cD}(\bm{p})$, the worker $k=\iota(y)$ belongs to $R$ and satisfies $r_k(y)\le r_k(x_k)$.
\end{enumerate}
Then, $(X,\bm{p}^{\varepsilon})$ is an equilibrium for every sufficiently small $\varepsilon>0$.
In particular, $(X,\bm{p})$ is not seller-optimal.
\end{lemma}

\begin{proof}
Since $\bm{p}^{\varepsilon}\ge\bm{p}\ge\bm{\beta}$ and $\bm{p}^{\varepsilon}$ agrees with $\bm{p}$ on unmatched workers, \eqref{eq:standard-equilibrium} reduces the claim to $X\in\hat{\cD}(\bm{p}^{\varepsilon})$.
The transfer on a trade $y=(k,j)$ rises by $\delta(y)=\varepsilon\lambda^{r_k(x_k)-r_k(y)}$ if $k\in R$ and by $\delta(y)=0$ otherwise, and we write $\delta(Y)=\sum_{y\in Y}\delta(y)$.
For small $\varepsilon$, every set of finite value outside $\hat{\cD}(\bm{p})$ remains strictly worse than $X$ for the buyers, while each $Y\in\hat{\cD}(\bm{p})$ loses $\delta(Y)$ from a common total utility.
Hence, it suffices that $X$ minimizes $\delta$ over $\hat{\cD}(\bm{p})$.
Since $\delta$ is additive and $\hat{\cD}(\bm{p})$ is a g-matroid, it suffices to check single-element additions, deletions, and exchanges at $X$ within $\hat{\cD}(\bm{p})$ (\Cref{sec:gs-preferences}).

An addition cannot decrease $\delta$ because $\delta\ge0$.
A deletion or an exchange at $x_i$ with $i\notin R$ cannot either, because $\delta(x_i)=0$.
For $i\in R$, condition~(i) excludes deleting $x_i$, and for an exchange $X-x_i+y$, we have $\delta(y)=\varepsilon\lambda^{r_k(x_k)-r_k(y)}\ge\varepsilon=\delta(x_i)$ by condition~(ii).
Finally, $\bm{p}^{\varepsilon}$ exceeds $\bm{p}$ on the nonempty set $R$, and $(X,\bm{p})$ is therefore not seller-optimal.
\end{proof}

The next lemma also covers the outside payoffs that arise in \Cref{sec:gda-analysis}.
Its condition on $\bm{p}$ is weaker than seller optimality, which \Cref{sec:rank-lattice} uses.
\begin{lemma}\label[lemma]{lem:payoff-lower}
Let $\bm{\beta}$ be an outside payoff vector in which each $\beta_i$ equals $\beta_i^0$ or $\tfrac12\lambda^{\rho_i}$ for some integer $\rho_i>r_i(\varnothing)$, and let $(X,\bm{p})$ be an equilibrium with outside payoffs $\bm{\beta}$ such that no equilibrium $(X,\bm{q})$ satisfies $\bm{q}\ge\bm{p}$ and $\bm{q}\ne\bm{p}$.
Then, $p_i\ge\lambda^{r_i(x_i)}$ for every matched worker $i$.
\end{lemma}

\begin{proof}
Since $p_i=t_{x_i}\lambda^{r_i(x_i)}$, it suffices to show that $t_{x_i}\ge1$ for every matched worker $i$.
Suppose, for a contradiction, that the set $S$ of matched workers $s$ with $t_{x_s}<1$ is nonempty.
We construct a nonempty set $R$ satisfying conditions~(i) and~(ii) of \Cref{lem:raising}, which then contradicts the assumption on $\bm{p}$.
By \eqref{eq:exchange-value} below, an exchange that keeps a demanded set preserves the adjusted transfer $t_x-\omega_{\iota(x)}$ modulo $\kappa$.
The set $R$ collects the matched workers whose adjusted transfers are congruent modulo $\kappa$ to that of a worker in $S$.

Let $R$ consist of the matched workers $i$ such that
\begin{align}
 t_{x_i}=\kappa z+\omega_i-\omega_s+t_{x_s}
 \quad\text{for some }s\in S\text{ and }z\in\mathbb{Z}.
 \label{eq:transfer-set}
\end{align}
Taking $i=s$ and $z=0$ shows $S\subseteq R$, hence $R\ne\varnothing$.
For $i\in R$, the number $\kappa z+\omega_i-\omega_s$ is an integer and $0<t_{x_s}<1$, which makes $t_{x_i}$ nonintegral.
Whenever $X-x_i+y\in\hat{\cD}(\bm{p})$, the two demanded sets attain the same total buyer utility, hence $t_{x_i}-t_y=\tilde{v}(X)-\tilde{v}(X-x_i+y)$.
By $\tilde{v}=\kappa v+\omega$, this reads
\begin{align}
 (t_{x_i}-\omega_i)-(t_y-\omega_{\iota(y)})=\kappa\bigl[v(X)-v(X-x_i+y)\bigr],
 \label{eq:exchange-value}
\end{align}
and the right side is an integer multiple of $\kappa$ because $v$ is integer-valued.

\medskip\noindent\textbf{Condition~(i).}
If $X-x_i\in\hat{\cD}(\bm{p})$ for some $i\in R$, then $t_{x_i}=\tilde{v}(X)-\tilde{v}(X-x_i)$ would be an integer.

\medskip\noindent\textbf{Condition~(ii).}
Let $X-x_i+y\in\hat{\cD}(\bm{p})$ with $i\in R$ and $y\in E\setminus X$, and write $k=\iota(y)$.
We must show $k\in R$ and $r_k(y)\le r_k(x_k)$.
With $s\in S$ as in \eqref{eq:transfer-set} for this $i$, that identity and \eqref{eq:exchange-value} give
\begin{align}
 t_y=\kappa z'+\omega_k-\omega_s+t_{x_s}
 \quad\text{for some }z'\in\mathbb{Z}.
 \label{eq:transfer-residue}
\end{align}
Since $X-x_i+y\in\hat{\cD}(\bm{p})$, we have $0<t_y\le\bar{v}$ by \Cref{lem:payoff-upper}.
If $t_y<1$, then the integer $\kappa z'+\omega_k-\omega_s=t_y-t_{x_s}$ lies in $(-1,1)$ because $0<t_{x_s}<1$, and hence vanishes.
Since $|\omega_k-\omega_s|<\kappa$, this forces $z'=0$ and $\omega_k=\omega_s$, and thus $k=s$ by the distinctness of the weights and $t_y=t_{x_s}$.

We next show that $k$ is matched.
Otherwise, $p_k=\beta_k$ and $t_y=\lambda^{-r_k(y)}\beta_k$, where $t_y\le\bar{v}<\lambda/2$.
If $\beta_k=\omega_k\lambda^{r_k(\varnothing)}$, then $t_y=\omega_k\lambda^{r_k(\varnothing)-r_k(y)}$ forces $r_k(y)\ge r_k(\varnothing)$, and $t_y$ equals $\omega_k$ when $r_k(y)=r_k(\varnothing)$ and is at most $\omega_k/\lambda<1$ otherwise.
If $\beta_k=\tfrac12\lambda^{\rho_k}$, then $t_y=\tfrac12\lambda^{\rho_k-r_k(y)}$ forces $\rho_k\le r_k(y)$, giving $t_y\le\tfrac12$.
Thus, $t_y=\omega_k$ or $0<t_y<1$.
The former makes $t_y$ an integer, contrary to \eqref{eq:transfer-residue}, and the latter implies $k=s$, whereas $s$ is matched.

Since $k$ is matched, $t_{x_k}\le\bar{v}$ by \Cref{lem:payoff-upper}, and $t_y=t_{x_k}\lambda^{r_k(x_k)-r_k(y)}$.
If $r_k(y)=r_k(x_k)$, then $t_{x_k}=t_y$, and \eqref{eq:transfer-residue} gives $k\in R$.
If $r_k(y)>r_k(x_k)$, then $t_y\le t_{x_k}/\lambda\le\bar{v}/\lambda<1$, hence $k=s$ and $t_y=t_{x_s}=t_{x_k}$, contradicting $t_y<t_{x_k}$.
In the remaining case, $r_k(y)<r_k(x_k)$, and $t_{x_k}\le t_y/\lambda\le\bar{v}/\lambda<1$ implies $k\in S\subseteq R$.
\end{proof}

We summarize these bounds for the baseline market.
\begin{lemma}\label[lemma]{lem:rank-support}
Every baseline equilibrium $(X,\bm{p})$ satisfies
\begin{align}
 p_i<\lambda^{r_i(X(i))+1}\quad\text{for every }i\in I.
 \label{eq:baseline-payoff-upper}
\end{align}
If $(X,\bm{p})$ is seller-optimal, then
\begin{align}
 \lambda^{r_i(X(i))}\le p_i<\lambda^{r_i(X(i))+1}\quad\text{for every }i\in I.
 \label{eq:tier-payoffs}
\end{align}
\end{lemma}
\begin{proof}
An unmatched worker has $p_i=\beta_i^0=\omega_i\lambda^{r_i(\varnothing)}<\lambda^{r_i(\varnothing)+1}$ because $\omega_i<\lambda$, and a matched worker has $p_i\le\bar{v}\lambda^{r_i(X(i))}<\lambda^{r_i(X(i))+1}$ by \Cref{lem:payoff-upper}, which gives~\eqref{eq:baseline-payoff-upper}.
The lower bound in~\eqref{eq:tier-payoffs} holds for unmatched workers because $\omega_i>1$, and for matched workers whenever no equilibrium $(X,\bm{q})$ satisfies $\bm{q}\ge\bm{p}$ and $\bm{q}\ne\bm{p}$, by \Cref{lem:payoff-lower}.
Seller-optimal payoffs satisfy this condition.
\end{proof}

\subsection{Proof of Theorem~\ref{thm:auxiliary-guarantees}}\label{sec:ordinal-guarantees}

We now translate the equilibrium conditions into the original preferences.
The payoff bounds of \Cref{lem:rank-support} identify each worker's tier, and the factor $\kappa$ makes differences in $v$ dominate the transfers.

\begin{proof}[Proof of \Cref{thm:auxiliary-guarantees}]
Fix any baseline equilibrium $(X,\bm{p})$, with transfers $t_{ij}=\lambda^{-r_i(j)}p_i$.
We show that $X$ is individually rational, weakly stable, and Pareto efficient, and then prove worker WGSP.

\medskip\noindent\textbf{Weak stability.}
If $X(i)\prec_i\varnothing$, then we have $p_i<\lambda^{r_i(X(i))+1}\le\lambda^{r_i(\varnothing)}\le\beta^0_i$ by \eqref{eq:baseline-payoff-upper}, contradicting $p_i\ge\beta^0_i$.
Thus, every worker finds her assignment acceptable.

We next rule out blocking coalitions.
As these include a firm's strict improvement by dropping workers, firm individual rationality follows as well.
Suppose a firm $j$ and a set $T\subseteq I$ block $X$.
Since $T\succ_jX(j)$ and $X(j)$ has finite value, we have $v_j(T)-v_j(X(j))\ge1$.
Every $i\in T\setminus X(j)$ has $j\succ_iX(i)$, and we then have $t_{ij}=\lambda^{-r_i(j)}p_i<\lambda^{r_i(X(i))+1-r_i(j)}\le1$ by \eqref{eq:baseline-payoff-upper}.
Buyer $j$ compares all subsets of $E_j$, including trades with workers assigned to other firms in $X$.
Therefore,
\begin{align*}
 \MoveEqLeft V_j(T\times\{j\},\bm{t}_j)-V_j(X_j,\bm{t}_j)\\
 &=\kappa\alpha_j[v_j(T)-v_j(X(j))]+\omega(T\times\{j\})-\omega(X_j)-\sum_{i\in T\setminus X(j)}t_{ij}+\sum_{i\in X(j)\setminus T}t_{ij}\\
 &\ge\kappa-\bar{\omega}-\sum_{i\in T\setminus X(j)}t_{ij}
 >\kappa-\bar{\omega}-n>0,
\end{align*}
contradicting $X_j\in D_j(\bm{t})$.
Thus, $X$ is weakly stable.

\medskip\noindent\textbf{Pareto efficiency.}
Suppose next that a matching $Y$ Pareto improves upon $X$.
Since $Y(j)\succsim_jX(j)\succsim_j\emptyset$ for every $j$, $v(Y)$ is finite and $v(Y)\ge v(X)$, strictly if a firm improves.
For any matching $Z$, $\sum_jV_j(Z_j,\bm{t}_j)=\kappa v(Z)+\sum_{i\in I\mathrel{:} Z(i)\ne\varnothing}(\omega_i-p_i\lambda^{-r_i(Z(i))})$.
Since $X_j\in D_j(\bm{t})$ for every $j$, we have
$\sum_j[V_j(Y_j,\bm{t}_j)-V_j(X_j,\bm{t}_j)]\le0$.
Expanding this difference and using $p_i\lambda^{-r_i(\varnothing)}=\omega_i$ for workers unmatched in $X$ yields
\begin{align*}
 0\ge\kappa[v(Y)-v(X)]+\sum_{i\in I}p_i\left(\lambda^{-r_i(X(i))}-\lambda^{-r_i(Y(i))}\right)+\sum_{i\in I\mathrel{:} Y(i)=\varnothing}\left(p_i\lambda^{-r_i(\varnothing)}-\omega_i\right).
\end{align*}
The three terms on the right are nonnegative because $v(Y)\ge v(X)$, $r_i(Y(i))\ge r_i(X(i))$, and $p_i\ge\beta_i^0=\omega_i\lambda^{r_i(\varnothing)}$.
A strict improvement by a firm makes the first term positive, and one by a worker makes her term in the sum over $i\in I$ positive, a contradiction.

\medskip\noindent\textbf{Worker WGSP.}
Finally, fix a deterministic mechanism selecting seller-optimal equilibrium matchings.
Pairing the selected matching with the seller-optimal payoffs at each report profile defines a deterministic seller-optimal mechanism on the profiles of the auxiliary market that arise from ordinal reports.
Extend it to all profiles in the domain of \Cref{thm:schlegel} by any seller-optimal selection.
The extension is seller WGSP by \Cref{thm:schlegel}, and so is its restriction to these profiles.
Let $X$ be the matching selected at the true profile and $\bm{p}$ its seller-optimal payoffs.
Suppose a coalition's joint misreport changes the selected matching to $X'$, supported by seller-optimal payoffs $\bm{p}'$ at the misreported profile, and strictly improves every member's assignment under her true preference.
Each member $i$ is then matched in $X'$, since $X(i)\succsim_i\varnothing$, and her transfer $t_i'=\lambda^{-r'_i(X'(i))}p'_i$ satisfies $t_i'\ge1$ by \Cref{lem:rank-support} at the misreported profile, where $r'_i$ is her reported rank.
Using her true rank $r_i$ and \eqref{eq:tier-payoffs} at the true profile gives $t_i'\lambda^{r_i(X'(i))}\ge\lambda^{r_i(X'(i))}\ge\lambda^{r_i(X(i))+1}>p_i$.
Every member therefore strictly increases her true utility in the auxiliary market, contradicting the seller WGSP of \Cref{thm:schlegel}.
\end{proof}

\section{An Equilibrium Characterization of Generalized Deferred Acceptance}\label{sec:gda-analysis}
In this section, we show that generalized deferred acceptance returns a seller-optimal equilibrium matching of the baseline market (\Cref{thm:terminal-equivalence}), which proves \Cref{thm:generalized-da} by \Cref{thm:auxiliary-guarantees}.
Although the procedure computes no transfers, we associate each of its rounds with the auxiliary market in which the workers' outside payoffs come from their current tiers, keeping all trades and buyer valuations fixed.
Later rounds lower the outside payoffs of the workers still unmatched, just as a descending auction lowers the price until a buyer accepts.
\Cref{sec:current-tier-clearing} characterizes the matchings retained in a round by welfare maximization, \Cref{sec:descending-payoffs} shows that they are the seller-optimal equilibrium matchings of that round's market, and \Cref{sec:terminal-optimality} passes from the terminal outside payoffs to the baseline ones.

\subsection{Welfare maximization within current tiers}\label{sec:current-tier-clearing}

Fix a round of generalized deferred acceptance, and let $F=(F_i)_{i\in I}$ be its tier profile after workers advance in Step~3.
The round retains the matchings in $\cM_{m+n}(F)$ and selects $\Psi(F)$ among them.
We describe $\cM_{m+n}(F)$ as the set of maximizers of a welfare problem over $\cM(F)$, and relate these maximizers to the equilibria of an auxiliary market whose outside payoffs depend on $F$.

The auxiliary market for the round uses the outside payoffs $\bm{\beta}(F)=(\beta_i(F_i))_{i\in I}$ given by
\begin{align*}
 \beta_i(F_i)=
 \begin{cases}
  \tfrac12\lambda^{\ell_i(F_i)}&\text{if an acceptable tier below }F_i\text{ remains},\\
  \beta_i^0&\text{otherwise},
 \end{cases}
\end{align*}
where $\ell_i(F_i)$ is the rank of tier $F_i$, that is, $\ell_i(F_i)=r_i(j)$ for all $j\in F_i$ when $F_i\ne\emptyset$, and $\ell_i(\emptyset)=r_i(\varnothing)$.

When an acceptable tier below $F_i$ remains, $F_i$ lies strictly above $\varnothing$, hence $\ell_i(F_i)>r_i(\varnothing)$.
Thus, \Cref{lem:payoff-lower} applies to $\bm{\beta}(F)$ with $\rho_i=\ell_i(F_i)$, and $\beta_i(F_i)\ge\tfrac12\lambda^{r_i(\varnothing)+1}>\omega_i\lambda^{r_i(\varnothing)}=\beta_i^0$ because $\lambda>2\bar{v}\ge2\omega_i$.
Since $\beta_i(F_i)=\beta_i^0$ otherwise, we have $\bm{\beta}(F)\ge\bm{\beta}^0$, which we use in the remaining subsections.
By \Cref{lem:payoff-upper,lem:payoff-lower}, at a seller-optimal equilibrium, an assignment in the current tier yields worker $i$ a payoff of at least $\lambda^{\ell_i(F_i)}$, while one in a lower tier yields at most $\bar{v}\lambda^{\ell_i(F_i)-1}$.
The outside payoff $\tfrac12\lambda^{\ell_i(F_i)}$ lies strictly between these values, and worker $i$ therefore accepts her current tier and rejects lower ones.

For each worker $i$, let $z_i(F_i)=\lambda^{-\ell_i(F_i)}\beta_i(F_i)$ be the transfer at which an assignment in her current tier gives her exactly her outside payoff.
We maximize the buyers' total valuation plus these transfers of the unmatched workers over the matchings within the current tiers:
\begin{align}
 \text{maximize}\quad \tilde{v}(\nu)+\sum_{i\in I\mathrel{:}\nu(i)=\varnothing}z_i(F_i)\qquad\text{subject to}\quad \nu\in\cM(F).
 \label{eq:round-welfare}
\end{align}

To reflect the firm index order, we now fix the weights $\alpha_j$, arbitrary so far, to be the positive integers $\alpha_{j_t}=M^{m-t}$ for $t=1,\dots,m$, where $U_j=\sum_{i\in I}\max\{0,v_j(\{i\})\}$ and $M=1+\max\{1,\max_{j\in J}U_j\}$.
The weights depend only on the GS representations $v_j$ and are therefore fixed before workers report.
Then, $v(X)=\sum_{t=1}^{m}M^{m-t}v_{j_t}(X(j_t))$.
Applying the exchange axiom to $S\in\dom(v_j)$ and $\emptyset$ yields $v_j(S)\le v_j(S\setminus\{i\})+v_j(\{i\})$ for each $i\in S$, and iterating shows that every finite value of $v_j$ is at most $U_j$.
Since $M$ exceeds every value a single firm can attain, maximizing $v$ over matchings in which every firm has a nonnegative value compares the firms lexicographically in their index order, just as their serial dictatorship does.
The proof of \Cref{prop:ordered-welfare} shows that no other matching is retained.

\begin{proposition}\label[proposition]{prop:ordered-welfare}
Fix a tier profile $F$ and consider the auxiliary market with outside payoffs $\bm{\beta}(F)$.
\begin{enumerate}[label=\textup{(\roman*)},leftmargin=*,topsep=0.2em,itemsep=0.2em]
 \item The maximizers of \eqref{eq:round-welfare} are exactly the matchings in $\cM_{m+n}(F)$.
 \item Let $(X,\bm{p})$ with $X\in\cM(F)$ be an equilibrium.
 Then, $\nu\in\cM(F)$ maximizes \eqref{eq:round-welfare} if and only if $(\nu,\bm{p})$ is an equilibrium.
\end{enumerate}
\end{proposition}

\begin{proof}
For (i), by $\tilde{v}=\kappa v+\omega$, the objective of \eqref{eq:round-welfare} can be written as
\begin{align}
 \kappa v(\nu)+\sum_{i\in I}z_i(F_i)
 +\sum_{i\in I\mathrel{:}\nu(i)\ne\varnothing}(\omega_i-z_i(F_i)).
 \label{eq:ordered-welfare}
\end{align}
Since $\emptyset\in\cM(F)$, the maximum in \eqref{eq:round-welfare} is finite.
The last sum in \eqref{eq:ordered-welfare} varies by at most $\bar{\omega}<\kappa$, while $\sum_{i\in I}z_i(F_i)$ is constant.
Thus, the objective first maximizes $v$.
Replacing a firm's set by $\emptyset$ keeps a matching in $\cM(F)$ and leaves the other firms' sets unchanged.
Hence, both this maximization and the firms' serial dictatorship retain only matchings in which every firm has a value in $\{0,\ldots,M-1\}$.
On such matchings, the digits of $v$ in base $M$ are the firms' values in index order, and maximizing $v$ therefore retains exactly $\cM_m(F)$.

Among these maximizers, matching worker $i$ changes the remaining terms by $\omega_i-z_i(F_i)$.
A worker with $F_i=\emptyset$ is unmatched throughout $\cM(F)$ and contributes a constant.
For a nonempty $F_i$ not strictly above $\varnothing$, the difference is zero.
When $F_i$ is strictly above $\varnothing$, we have $z_i(F_i)<1$, and the difference exceeds $\omega_i-1$.
This exceeds the later workers' differences because $\omega_{i_q}-1>\sum_{r>q}\omega_{i_r}$.
Maximizing these terms therefore reproduces the workers' choices in index order, giving $\cM_{m+n}(F)$.

For (ii), suppose $(X,\bm{p})$ with $X\in\cM(F)$ is an equilibrium.
Put $q_i=\lambda^{-\ell_i(F_i)}p_i$, the transfer on each of her current-tier trades, abbreviate $z_i(F_i)$ to $z_i$, and let $u$ be the sum of the buyers' maximum utilities at $\bm{p}$. 
The value $u+\sum_{i\in I}q_i$ does not depend on $\nu$, and every $\nu\in\cM(F)$ with finite value satisfies
\begin{align*}
 u+\sum_{i\in I}q_i-\biggl[\tilde{v}(\nu)+\sum_{i\in I\mathrel{:}\nu(i)=\varnothing}z_i\biggr]=\biggl[u-\tilde{v}(\nu)+\sum_{i\in I\mathrel{:}\nu(i)\ne\varnothing}q_i\biggr]+\sum_{i\in I\mathrel{:}\nu(i)=\varnothing}(q_i-z_i).
\end{align*}
Both terms on the right are nonnegative, since
$X\in\hat{\cD}(\bm{p})$ and
$\bm{p}\ge\bm{\beta}(F)$, respectively.
Since $(X,\bm{p})$ is an equilibrium, both terms vanish at $X$.
Thus, $\nu$ maximizes the objective in \eqref{eq:round-welfare} if and only if $\nu\in\hat{\cD}(\bm{p})$ and $p_i=\beta_i(F_i)$ whenever $\nu(i)=\varnothing$.
By \eqref{eq:standard-equilibrium}, these two conditions and $\bm{p}\ge\bm{\beta}(F)$ hold if and only if $(\nu,\bm{p})$ is an equilibrium.
\end{proof}

Equilibrium itself, however, is not restricted to the domain $\cM(F)$ of \eqref{eq:round-welfare}, since each buyer $j$ compares all subsets of $E_j$, including trades outside the current tiers.
\Cref{lem:round-maximizers} shows that, at every tier profile reached after a round, every seller-optimal equilibrium matching nevertheless lies in $\cM(F)$.

\subsection{Equilibrium matchings in each round}\label{sec:descending-payoffs}

For the retained matchings after each round to coincide with the seller-optimal equilibrium matchings of that round's market, two kinds of pairs must be excluded, namely those below the current tiers and those in previously left tiers.
The outside payoffs of \Cref{sec:current-tier-clearing} exclude the former, since such a pair would require a transfer above $\bar{v}$, contrary to \Cref{lem:payoff-upper}.
For the latter, a worker who leaves a tier of rank $k$ has payoff equal to her outside payoff $\tfrac12\lambda^k$, below the level $\lambda^k$ that \Cref{lem:payoff-lower} requires for assignment at rank $k$.
It remains to show that her payoff never returns to this level.
After the first round, $\bm{\beta}(F)$ weakly decreases whenever unmatched workers move to lower tiers.
This includes a worker $i$ reaching her last acceptable tier, since the tier she leaves has outside payoff above $\beta_i^0$.
By \Cref{thm:schlegel}, the equilibrium payoff vectors with outside payoffs $\bm{\beta}$ have a greatest element, the payoff vector of the seller-optimal equilibria, which we denote by $\bm{p}^+(\bm{\beta})$.
It therefore suffices that $\bm{p}^+(\bm{\beta})$ is monotone in the outside payoffs, which the following lemma shows.

\begin{lemma}\label[lemma]{lem:payoff-monotone}
If $\bm{\beta}\ge\bm{\beta}'$, then $\bm{p}^+(\bm{\beta})\ge\bm{p}^+(\bm{\beta}')$ coordinatewise.
\end{lemma}


The next lemma combines two equilibria at the coordinatewise maximum or minimum of their payoffs.
We use the maximum case to prove \Cref{lem:payoff-monotone} and the minimum case in \Cref{cor:rank-lattice}.
Its proof uses the following two demand comparisons from Schlegel~\cite[Lemma~4]{Schlegel2022}, with the roles of $X$ and $Y$ exchanged in the second.
Write $\bm{a}\vee\bm{b}$ and $\bm{a}\wedge\bm{b}$ for the coordinatewise maximum and minimum of real vectors $\bm{a}$ and $\bm{b}$ with the same index set.
For any buyer $j$, transfer vectors $\bm{t},\bm{t}'\in\mathbb{R}^E$, and sets $X\in\cD_j(\bm{t})$ and $Y\in\cD_j(\bm{t}')$, there exist $Z\in\cD_j(\bm{t}\vee\bm{t}')$ and $Z'\in\cD_j(\bm{t}\wedge\bm{t}')$ such~that
\begin{align}
  \{x\in X\mid t_x\ge t'_x\}\cup\{y\in Y\mid t'_y>t_y\}\subseteq Z,\qquad\quad &|Z|\le|X|,
  \label{eq:demand-comparison}\\
  Z'\subseteq\{y\in Y\mid t_y\ge t'_y\}\cup\{x\in X\mid t'_x>t_x\},\qquad\quad &|Z'|\ge|Y|.
  \label{eq:demand-comparison-min}
\end{align}

\begin{lemma}\label[lemma]{lem:equilibrium-splicing}
Let $(X,\bm{p})$ and $(Y,\bm{p}')$ be equilibria of the auxiliary market with outside payoffs $\bm{\beta}$ and $\bm{\beta}'$, respectively, where $\bm{\beta}\ge\bm{\beta}'$.
Let $Z^+$ give each worker $i$ her assignment in $X$ if $p_i\ge p'_i$ and in $Y$ otherwise, and let $Z^-$ give her the other one.
Then, $(Z^+,\bm{p}\vee\bm{p}')$ is an equilibrium with outside payoffs $\bm{\beta}$.
If $\bm{\beta}=\bm{\beta}'$, then $(Z^-,\bm{p}\wedge\bm{p}')$ is an equilibrium as well.
\end{lemma}
\begin{proof}
Write $H=\{i\in I\mid p_i\ge p'_i\}$ and $L=I\setminus H$, and write $X_H$ for the trades in $X$ of the workers in $H$, with $X_L$, $Y_H$, and $Y_L$ defined likewise.
Then, $Z^+=X_H\cup Y_L$ and $Z^-=Y_H\cup X_L$.
Every $i\in L$ satisfies $p'_i>p_i\ge\beta_i\ge\beta'_i$ and is therefore matched in $Y$.
Hence, $|Y_L|=|L|\ge|X_L|$ and $|Z^+|=|X_H|+|Y_L|\ge|X|$. 

Since the transfer for a trade $(i,j)$ is $\lambda^{-r_i(j)}p_i$, the transfer vectors at $\bm{p}\vee\bm{p}'$ and at $\bm{p}\wedge\bm{p}'$ are the coordinatewise maximum and minimum of those at $\bm{p}$ and at $\bm{p}'$.
Applying \eqref{eq:demand-comparison} to each buyer and taking the union of the resulting sets produces a set in $\hat{\cD}(\bm{p}\vee\bm{p}')$ that contains $Z^+$ and has at most $|X|$ elements.
Since $|Z^+|\ge|X|$, that set is $Z^+$.
A worker unmatched in $Z^+$ lies in $H$ and is unmatched in $X$, hence $(\bm{p}\vee\bm{p}')_i=p_i=\beta_i$.
With $\bm{p}\vee\bm{p}'\ge\bm{p}\ge\bm{\beta}$, \eqref{eq:standard-equilibrium} makes $(Z^+,\bm{p}\vee\bm{p}')$ an equilibrium with outside payoffs $\bm{\beta}$.

Suppose now that $\bm{\beta}=\bm{\beta}'$.
Then, $|Z^-|=|Y_H|+|X_L|\le|Y_H|+|Y_L|=|Y|$.
Similarly, \eqref{eq:demand-comparison-min} yields a set in
$\hat{\cD}(\bm{p}\wedge\bm{p}')$ contained in $Z^-$
with at least $|Y|$ elements, which must therefore equal $Z^-$.
Each worker takes her assignment in $Z^-$ from the equilibrium with the smaller payoff, hence $(\bm{p}\wedge\bm{p}')_i=\beta_i$ whenever she is unmatched in $Z^-$.
With $\bm{p}\wedge\bm{p}'\ge\bm{\beta}$, the pair $(Z^-,\bm{p}\wedge\bm{p}')$ is an equilibrium as well.
\end{proof}

\begin{proof}[Proof of \Cref{lem:payoff-monotone}]
Let $X$ and $Y$ be matchings supporting $\bm{p}^+(\bm{\beta})$ and $\bm{p}^+(\bm{\beta}')$, respectively.
By \Cref{lem:equilibrium-splicing}, $\bm{p}^+(\bm{\beta})\vee\bm{p}^+(\bm{\beta}')$ is an equilibrium payoff vector with outside payoffs $\bm{\beta}$.
As $\bm{p}^+(\bm{\beta})$ is the greatest such vector, we have $\bm{p}^+(\bm{\beta})\ge\bm{p}^+(\bm{\beta})\vee\bm{p}^+(\bm{\beta}')\ge\bm{p}^+(\bm{\beta}')$.
\end{proof}

We now use payoff monotonicity to exclude pairs in previously left tiers.

\begin{lemma}\label[lemma]{lem:round-maximizers}
Let $F$ be a tier profile reached after a round.
Then, $X\in\cM_{m+n}(F)$ if and only if $(X,\bm{p}^+(\bm{\beta}(F)))$ is an equilibrium of the auxiliary market with outside payoffs $\bm{\beta}(F)$.
\end{lemma}

\begin{proof}
Write $F^t$ for the tier profile after round $t$, $\bm{\beta}^t=\bm{\beta}(F^t)$, and $\bm{p}^t=\bm{p}^+(\bm{\beta}^t)$.
We prove by induction that $p_i^t<\lambda^{r_i(j)}$ whenever worker $i$ has left the tier containing $j$.
By \Cref{lem:payoff-upper,lem:payoff-lower}, we have $1\le t_{x_i}\le\bar{v}$ for each matched worker $i$ at a seller-optimal equilibrium with outside payoffs $\bm{\beta}^t$.
The transfer upper bound also holds at any equilibrium.
No equilibrium assigns a pair $(i,j)$ below worker $i$'s current tier, since such a pair would require a transfer of at least $\lambda/2>\bar{v}$ when $\beta_i=\tfrac12\lambda^{\ell_i(F_i)}$ and at least $\omega_i\lambda>\bar{v}$ when $\beta_i=\beta_i^0$.
Both contradict the transfer upper bound.

\medskip\noindent\textbf{First round.}
After the first round, every worker with an acceptable firm is at her highest acceptable tier.
No tier has been left, and the exclusion of pairs below the current tiers implies that every seller-optimal equilibrium matching belongs to $\cM(F^1)$.
By \Cref{thm:schlegel}, a seller-optimal equilibrium exists, and its matching lies in $\cM(F^1)$.
By \Cref{prop:ordered-welfare}, we then have the claimed equality.

\medskip\noindent\textbf{Induction step.}
Suppose the claims hold after round $t$.
Since $\bm{\beta}^{t+1}\le\bm{\beta}^t$, we have $\bm{p}^{t+1}\le\bm{p}^t$ by \Cref{lem:payoff-monotone}, preserving the invariant for previously left tiers.
A worker $i$ newly leaving a tier of rank $k$ was unmatched in $\Psi(F^t)$ and had a next acceptable tier.
By the induction hypothesis, $(\Psi(F^t),\bm{p}^t)$ is an equilibrium, and hence $p_i^{t+1}\le p_i^t=\beta_i^t=\tfrac12\lambda^k<\lambda^k$, the invariant for the newly left tier.
Assigning a worker $i$ to a firm $j$ in a previously left tier at a seller-optimal equilibrium would require $p_i^{t+1}\ge\lambda^{r_i(j)}$ by the transfer lower bound, contradicting the invariant.
Together with the exclusion of pairs below the current tiers, this puts every seller-optimal equilibrium matching in $\cM(F^{t+1})$.
Applying \Cref{prop:ordered-welfare} again completes the induction.
\end{proof}

\subsection{Seller optimality at termination}\label{sec:terminal-optimality}

\Cref{lem:round-maximizers} identifies the matchings retained at termination with the seller-optimal equilibrium matchings of the auxiliary market with outside payoffs $\bm{\beta}(F^T)$, where $F^T$ is the terminal tier profile.
To apply \Cref{thm:auxiliary-guarantees}, we show that replacing these outside payoffs by $\bm{\beta}^0$ leaves the seller-optimal payoff vector and its equilibrium matchings unchanged.

\begin{theorem}\label[theorem]{thm:terminal-equivalence}
Let $F^T$ be the terminal tier profile of generalized deferred acceptance, and let the firm weights be $\alpha_{j_t}=M^{m-t}$ for $t=1,\dots,m$.
Then, $X\in\cM_{m+n}(F^T)$ if and only if $(X,\bm{p}^+(\bm{\beta}^0))$ is an equilibrium of the baseline market.
\end{theorem}

\begin{proof}
If no worker has an acceptable firm, both $\cM_{m+n}(F^T)$ and the set of baseline equilibrium matchings consist only of the empty matching, the latter by individual rationality in \Cref{thm:auxiliary-guarantees}.
Otherwise, $F^T$ is reached after a round, and \Cref{lem:round-maximizers} applies.

Write $\bm{\beta}^T=\bm{\beta}(F^T)$ and $\bm{p}^T=\bm{p}^+(\bm{\beta}^T)$.
Then, $(\Psi(F^T),\bm{p}^T)$ is an equilibrium by \Cref{lem:round-maximizers}.
If $\beta_i^T>\beta_i^0$, worker $i$ has a next acceptable tier and is matched in $\Psi(F^T)$ by Step~2.
By \Cref{lem:payoff-lower}, $p_i^T\ge\lambda^{\ell_i(F_i^T)}=2\beta_i^T>\beta_i^T$.
Since $\bm{\beta}^0\le\bm{\beta}^T\le\bm{p}^T$ and $\beta_i^T=\beta_i^0$ for every other worker, $p_i^T=\beta_i^T$ if and only if $p_i^T=\beta_i^0$.
Buyer demand is independent of the outside payoffs.
Thus, by \eqref{eq:standard-equilibrium}, a pair $(X,\bm{p}^T)$ is an equilibrium of the auxiliary market with outside payoffs $\bm{\beta}^T$ if and only if it is an equilibrium of the baseline market.
In particular, $\bm{p}^T$ is a baseline equilibrium payoff vector, implying $\bm{p}^T\le\bm{p}^+(\bm{\beta}^0)$.
\Cref{lem:payoff-monotone} gives the reverse inequality, hence $\bm{p}^T=\bm{p}^+(\bm{\beta}^0)$.
Together with \Cref{lem:round-maximizers}, these observations show that $X\in\cM_{m+n}(F^T)$ if and only if $(X,\bm{p}^+(\bm{\beta}^0))$ is an equilibrium of the baseline market.
\end{proof}

Combining \Cref{thm:terminal-equivalence} with the guarantees of the auxiliary market in \Cref{thm:auxiliary-guarantees}, we now prove our main result.

\begin{proof}[Proof of \Cref{thm:generalized-da}]
By \Cref{thm:terminal-equivalence}, generalized deferred acceptance selects a seller-optimal equilibrium matching of the baseline market at every report profile, with all representations, weights, and constants fixed before workers report.
By \Cref{thm:auxiliary-guarantees}, the procedure is worker WGSP and selects a Pareto-stable matching.
The running time follows from \Cref{prop:single-vmi}.
\end{proof}

\section{Further Properties of Generalized Deferred Acceptance}\label{sec:further-properties}
This section derives three further properties of generalized deferred acceptance from its equilibrium characterization.
Throughout, we fix the public firm preferences and the index orders.
For each profile $\succsim_I$ of worker weak orders, write $\mathrm{GDA}(\succsim_I)$ for the matching returned by generalized deferred acceptance and $\mathrm{GDA}_i(\succsim_I)$ for worker $i$'s assignment in it.

\subsection{Coalitional stability}\label{sec:coalition-stability}
Weak stability rules out deviations by a single firm together with some workers.
The matching returned by generalized deferred acceptance also resists deviations by coalitions of several firms.

\begin{proposition}\label[proposition]{prop:coalition-stability}
Let $X=\mathrm{GDA}(\succsim_I)$ for some profile $\succsim_I$.
There exist no set $J'\subseteq J$ of firms and matching $Y'$ with $Y'(i)\in J'\cup\{\varnothing\}$ for every worker $i$ such that
\begin{enumerate}[label=\textup{(\roman*)},leftmargin=*,topsep=0.2em,itemsep=0.2em]
 \item $Y'(j)\succsim_jX(j)$ for every $j\in J'$, with $Y'(j)\succ_jX(j)$ for some $j\in J'$, and
 \item $Y'(i)\succsim_iX(i)$ for every worker $i$ with $Y'(i)\ne\varnothing$, with $Y'(i)\succ_iX(i)$ whenever $X(i)\in J\setminus J'$.
\end{enumerate}
\end{proposition}
\begin{proof}
By \Cref{thm:terminal-equivalence}, $X$ is a baseline equilibrium matching.
Fix supporting payoffs $\bm{p}$ with transfers $t_{ij}=\lambda^{-r_i(j)}p_i$, and suppose that $J'$ and $Y'$ are as in the statement.
Let $X'=\{(i,j)\in X\mid j\in J'\}$.
Since $X_j\in\cD_j(\bm{t})$ for every buyer $j$, replacing $X'$ by $Y'$ cannot increase the total utility of the buyers in $J'$.
This replacement changes their total utility by
\begin{align*}
 \kappa\sum_{j\in J'}\alpha_j\bigl[v_j(Y'(j))-v_j(X(j))\bigr]+\sum_{i\in I}c_i,
\end{align*}
where $c_i$ collects worker $i$'s terms $\omega_i-t_{ij}$, with a plus sign for her pair in $Y'$ and a minus sign for her pair in $X'$.
The first term is at least $\kappa$, since the values are integers and every firm in $J'$ weakly improves, at least one strictly.
Let $H=\{i\in I\mid X(i)\in J'\}$ and $K=\{i\in I\mid Y'(i)\ne\varnothing\}$, and note that $c_i=0$ for $i\notin H\cup K$.
For $i\in H\cap K$, we have $r_i(Y'(i))\ge r_i(X(i))$, and thus $t_{i,Y'(i)}\le t_{i,X(i)}$ and $c_i\ge0$.
For $i\in K\setminus H$ with $X(i)=\varnothing$, we have $p_i=\beta_i^0$ and $r_i(Y'(i))\ge r_i(\varnothing)$, and thus $t_{i,Y'(i)}\le\omega_i$ and $c_i\ge0$.
For $i\in K\setminus H$ with $X(i)\in J\setminus J'$, we have $t_{i,Y'(i)}<1<\omega_i$ by condition~(ii) and \eqref{eq:baseline-payoff-upper}, and thus $c_i>0$.
For $i\in H\setminus K$, we have $c_i=t_{i,X(i)}-\omega_i\ge-\omega_i$.
The total change is therefore at least $\kappa-\bar{\omega}>0$, a contradiction.
\end{proof}

Pareto-stable matchings need not have this property.
Consider a firm $j$ with quota one and $\{b\}\succ_j\{a\}\succ_j\varnothing$, a worker $a$ with $j\succ_a\varnothing$, and a worker $b$ with $j\sim_b\varnothing$.
The matching $\{(a,j)\}$ is Pareto-stable, whereas $J'=\{j\}$ and $Y'=\{(b,j)\}$ satisfy (i) and (ii).
Generalized deferred acceptance returns $\{(b,j)\}$.

\subsection{Lattice structure of assigned-tier profiles}\label{sec:rank-lattice}
We next describe the workers' assigned tiers across all baseline equilibrium matchings.
Write $\bm{r}(X)=(r_i(X(i)))_{i\in I}$ for the rank vector of a matching $X$, ordered coordinatewise, with larger ranks preferred by workers.
The baseline equilibrium matchings, and hence the lattice below, depend on the index orders through the weights $\alpha_j$ and $\omega_i$.

A supporting payoff vector of a baseline equilibrium matching $X$ is a \emph{rank-identifying support} of $X$ if it satisfies \eqref{eq:tier-payoffs}.
Every baseline equilibrium matching $X$ has one.
Indeed, the set $P_X=\{\bm{p}\mathrel{:}(X,\bm{p})\text{ is a baseline equilibrium}\}$ is nonempty and closed by \eqref{eq:standard-equilibrium}, whose demand conditions are finitely many weak linear inequalities, and bounded by \eqref{eq:baseline-payoff-upper}.
A vector $\bm{p}\in P_X$ maximizing $\sum_ip_i$ admits no coordinatewise increase in $P_X$, and the proof of \Cref{lem:rank-support} shows that it satisfies \eqref{eq:tier-payoffs}.

\begin{corollary}\label[corollary]{cor:rank-lattice}
Fix a profile $\succsim_I$.
The rank vectors of the baseline equilibrium matchings form a distributive lattice under the coordinatewise maximum and minimum, and $\bm{r}(\mathrm{GDA}(\succsim_I))$ is its greatest element.
\end{corollary}
\begin{proof}
Let $X$ and $X'$ be baseline equilibrium matchings with rank-identifying supports $\bm{p}$ and $\bm{p}'$.
For each worker $i$, the intervals in \eqref{eq:tier-payoffs} for different ranks are disjoint and ordered by rank.
Hence, $p_i\ge p'_i$ implies $r_i(X(i))\ge r_i(X'(i))$, and $p_i<p'_i$ implies $r_i(X(i))\le r_i(X'(i))$.
By \Cref{lem:equilibrium-splicing} with $\bm{\beta}=\bm{\beta}'=\bm{\beta}^0$, the matchings $Z^+$ and $Z^-$ are baseline equilibrium matchings, supported by $\bm{p}\vee\bm{p}'$ and $\bm{p}\wedge\bm{p}'$, respectively.
In $Z^+$, each worker receives her assignment from the equilibrium with the larger payoff, whose rank is the higher one by the implications above.
Thus, $\bm{r}(Z^+)=\bm{r}(X)\vee\bm{r}(X')$, and likewise $\bm{r}(Z^-)=\bm{r}(X)\wedge\bm{r}(X')$.
The rank vectors of the baseline equilibrium matchings therefore form a sublattice of $\mathbb{Z}^I$ under the coordinatewise order.
This sublattice is distributive because $\mathbb{Z}^I$ is.

Since $\mathrm{GDA}(\succsim_I)=\Psi(F^T)\in\cM_{m+n}(F^T)$, \Cref{thm:terminal-equivalence} shows that $\mathrm{GDA}(\succsim_I)$ is supported by $\bm{p}^+(\bm{\beta}^0)$.
This vector is a rank-identifying support by \Cref{lem:rank-support}, and it is the greatest equilibrium payoff vector of the baseline market by definition.
By the implications above, every baseline equilibrium matching $X'$ satisfies $\bm{r}(\mathrm{GDA}(\succsim_I))\ge\bm{r}(X')$.
\end{proof}

Under strict preferences, stable matchings form a lattice, and worker-proposing deferred acceptance selects the stable matching that every worker weakly prefers to all others~\cite{GaleShapley1962,HatfieldMilgrom2005}.
With indifferences, weakly stable matchings need not have such a common best element for workers, and \Cref{cor:rank-lattice} shows that the baseline equilibrium matchings restore this structure for assigned tiers.

\subsection{Worker-Pareto undominance}\label{sec:undominance}
Although generalized deferred acceptance is best for workers among baseline equilibrium matchings by \Cref{cor:rank-lattice}, it need not be worker-optimal among weakly stable matchings.
In the example of \Cref{sec:coalition-stability}, it returns $\{(b,j)\}$, whereas both workers weakly prefer the weakly stable matching $\{(a,j)\}$ and worker $a$ strictly.
Such Pareto inefficiency for workers among weakly stable matchings is inevitable for strategy-proof mechanisms that always select weakly stable matchings~\cite{ErdilErgin2008}.

We therefore ask whether a strategy-proof mechanism can at least improve on generalized deferred acceptance for workers.
Abdulkadiro\u{g}lu, Pathak, and Roth~\cite{AbdulkadirogluPathakRoth2009} answer this question for school choice with strict student preferences.
Let $\mathrm{DA}^\tau$ denote student-proposing deferred acceptance after the ties in school priorities are broken by a fixed rule $\tau$.
They show that no strategy-proof mechanism, stable or not, dominates $\mathrm{DA}^\tau$, that is, makes every student weakly better off than under $\mathrm{DA}^\tau$ at every profile and some student strictly better off at some profile.
Hirata and Kasuya~\cite{HirataKasuya2017} extend this result to many-to-one matching with contracts under strict doctor preferences.

In generalized deferred acceptance, the fixed index orders of firms and workers play the role of $\tau$.
Unlike in these results, workers may be indifferent, and a mechanism may then move a worker to another firm she finds equally good.
We require such moves to respect the firms' index order.
For a matching $X$, let $F^X$ be the profile with $F^X_i=\{j\in J\mid j\sim_iX(i)\}$, which places each worker in the indifference class of her assignment.
A matching $X$ satisfies \emph{tierwise firm priority} if $X\in\cM_m(F^X)$, that is, if the firms' serial dictatorship over $\cM(F^X)$ retains $X$.
Since $\cM(F^X)$ contains every matching obtained from $X$ by removing workers, tierwise firm priority implies individual rationality for firms.
A mechanism satisfies tierwise firm priority if its outcome does at every profile.

We first show that, at each profile, workers can improve on generalized deferred acceptance only by lowering the weighted sum $v$ of the firms' representations.
\begin{lemma}\label[lemma]{lem:lex-efficiency}
Let $X=\mathrm{GDA}(\succsim_I)$ for some profile $\succsim_I$, and let $Y$ be a matching with $v(Y)\ge v(X)$ for the weights of \Cref{sec:current-tier-clearing}.
If $Y(i)\succsim_iX(i)$ for every worker $i\in I$, then $Y(i)\sim_iX(i)$ for every worker $i\in I$.
\end{lemma}
\begin{proof}
By \Cref{thm:terminal-equivalence}, $X$ is a baseline equilibrium matching.
The argument for Pareto efficiency in the proof of \Cref{thm:auxiliary-guarantees} uses the firms' preferences only through $v(Y)\ge v(X)$, and it therefore shows that no worker strictly prefers $Y(i)$ to $X(i)$.
\end{proof}

Combined with strategy-proofness, \Cref{lem:lex-efficiency} yields the main result of this section.
\begin{theorem}\label[theorem]{thm:undominance}
Let $\Phi$ be a deterministic mechanism that is worker strategy-proof and satisfies tierwise firm priority.
If $\Phi_i(\succsim_I)\succsim_i\mathrm{GDA}_i(\succsim_I)$ for every profile $\succsim_I$ and every worker $i\in I$, then $\Phi_i(\succsim_I)\sim_i\mathrm{GDA}_i(\succsim_I)$ for every profile $\succsim_I$ and every worker $i\in I$.
\end{theorem}
\begin{proof}
We first show that $\mathrm{GDA}$ satisfies tierwise firm priority.
Fix a profile $\succsim_I$, and write $X=\mathrm{GDA}(\succsim_I)=\Psi(F^T)$, where $F^T$ is the terminal tier profile.
At termination, every matched worker's current tier is the indifference class of her assignment, and every unmatched worker has reached her last acceptable tier, if she has one, which contains every firm she finds indifferent to remaining unmatched.
Hence, $F^X_i\subseteq F^T_i$ for every worker $i$, and thus $\cM(F^X)\subseteq\cM(F^T)$.
Since $X$ lies in $\cM(F^X)$ and in $\cM_m(F^T)$, the set of matchings lexicographically maximal for the firms over $\cM(F^T)$, it also lies in $\cM_m(F^X)$.

Now suppose that the conclusion fails.
Among the profiles at which some worker strictly prefers her assignment under $\Phi$ to that under $\mathrm{GDA}$, choose $\succsim_I$ with the fewest acceptable pairs, that is, pairs $(i,j)$ with $j\succsim_i\varnothing$, and write $X=\mathrm{GDA}(\succsim_I)$ and $Y=\Phi(\succsim_I)$.
Suppose that some worker $i$ satisfies $Y(i)\succ_iX(i)\ne\varnothing$.
Let $\succsim'_i$ keep $i$'s ranking of firms and place $\varnothing$ alone just above the tier of $X(i)$, so that her acceptable firms are those she strictly prefers to $X(i)$, and put $\succsim'_I=(\succsim'_i,\succsim_{I\setminus\{i\}})$.
If $\mathrm{GDA}$ matched $i$ at $\succsim'_I$, she would gain at $\succsim_I$ by reporting $\succsim'_i$, contradicting the strategy-proofness of $\mathrm{GDA}$ (\Cref{thm:generalized-da}).
Hence, $\mathrm{GDA}_i(\succsim'_I)=\varnothing$.
By the strategy-proofness of $\Phi$ at $\succsim'_I$, we have $\Phi_i(\succsim'_I)\succsim'_iY(i)\succ'_i\varnothing$.
Thus, worker $i$ strictly prefers her assignment under $\Phi$ to that under $\mathrm{GDA}$ at $\succsim'_I$, which has fewer acceptable pairs, contradicting the choice of $\succsim_I$.

Hence, every worker who strictly prefers $Y(i)$ to $X(i)$ is unmatched in $X$, and every worker matched in $X$ is indifferent between $X(i)$ and $Y(i)$.
Thus, $X\in\cM(F^Y)$.
By tierwise firm priority, $Y\in\cM_m(F^Y)$, and the argument in the proof of \Cref{prop:ordered-welfare}~(i) shows that $\cM_m(F^Y)$ consists of the maximizers of $v$ over $\cM(F^Y)$.
Hence, $v(Y)\ge v(X)$, and \Cref{lem:lex-efficiency} shows that every worker is indifferent between $X(i)$ and $Y(i)$, contradicting the choice of $\succsim_I$.
\end{proof}

\Cref{thm:undominance} requires neither stability of $\Phi$ nor any condition involving the worker index order.
Its proof also applies to mechanisms defined only on strict worker preferences, since the truncation of a strict order is strict.
On this domain, $\cM(F^X)$ consists of the subsets of $X$, and tierwise firm priority reduces to individual rationality for firms.
With positive additive scores and quotas that rank workers strictly by the school priorities broken by $\tau$, generalized deferred acceptance coincides with $\mathrm{DA}^\tau$, and we recover the result of Abdulkadiro\u{g}lu, Pathak, and Roth~\cite{AbdulkadirogluPathakRoth2009}.

\section{Concluding Remarks}\label{sec:discussion}

In this paper, we construct generalized deferred acceptance and establish its Pareto stability and worker WGSP under public GS-representable firm preferences.
We also show that its outcome resists deviations by coalitions of several firms, attains the greatest profile of assigned tiers among equilibrium matchings, and is not uniformly Pareto dominated for workers by any strategy-proof mechanism respecting the firms' index order within tiers.

Pareto stability and worker WGSP hold not only for serial dictatorship but also when the firms' stage of each round maximizes a positive-integer weighted sum of their GS representations, fixed before workers report, while retaining all maximizers and applying the same ordered worker choices and final selection.
More generally, the same argument applies when the firms' stage instead maximizes a single integer-valued GS function $v$ on worker--firm pairs with $v(\emptyset)=0$ that is monotone in the firms' preferences.
Monotonicity means that, for all sets $X$ and $Y$ of pairs with $v(Y)>-\infty$, $v(X)\ge v(Y)$ whenever every firm $j$ weakly prefers $X(j)$ to $Y(j)$, with strict inequality whenever some firm $j$ strictly prefers $X(j)$ to $Y(j)$.
In this case, the firms act as a single buyer with valuation $\kappa v+\omega$, the proof of \Cref{thm:auxiliary-guarantees} uses the monotonicity of $v$ in place of each firm's demand, comparing $X$ with $(X\setminus X_j)\cup(T\times\{j\})$ for a blocking coalition $(j,T)$, and \Cref{sec:gda-analysis} applies with the firms' stage of each round maximizing $v$.
\Cref{thm:undominance} also extends, with tierwise firm priority defined through the firms' stage of the variant in place of their serial dictatorship.

For instance, when each firm compares sets of workers by the preorder generated by responsiveness from a weak ranking of individual workers, an assignment valuation whose slot weights decrease lexicographically in the round-robin order satisfies this condition, because an optimal assignment places better workers in slots with larger weights.
If every worker finds all firms acceptable and, writing $n=km+r$ with $0\le r<m$, firms $j_1,\ldots,j_r$ have $k+1$ slots and the others have $k$ slots, maximizing it reproduces the round-robin choice of Igarashi et al.~\cite{IgarashiEtAl2026}, since the weights make it fill the slots one at a time in that order, each with a best remaining worker.

For our procedure, GS-representability cannot be weakened to the induced choice correspondences.
A choice correspondence records choices from sets of available workers, not comparisons among the bundles left by earlier firms' optimality.
In \Cref{app:pi-lad-failure}, every firm's choice correspondence is induced by an additive GS function on a partition-matroid domain, yet the procedure returns an unstable matching even though a Pareto-stable one exists.

A natural extension would evaluate firm welfare using only the revealed preference preorders generated by choice correspondences, following the approach of Alkan and Gale~\cite{AlkanGale2003}.
Natural conditions on these correspondences are path independence and the law of aggregate demand in the sense of Bando, Imamura, and Kawase~\cite{BIK2025}.
When workers may be indifferent, it remains open whether these conditions suffice for some mechanism to guarantee weak stability, Pareto efficiency with respect to these preorders and the workers' preferences, and worker WGSP.

\section*{Acknowledgments}
This work was supported by JSPS KAKENHI Grant Numbers JP25K00137 and JP26K21174, and JST ERATO Grant Number JPMJER2301.

We used OpenAI's GPT and Anthropic's Claude to assist with developing research ideas, constructing and checking proofs, and drafting and revising the manuscript.
The authors take full responsibility for the content of this paper, including all mathematical claims and proofs.

\bibliographystyle{plain}
\bibliography{references}

\appendix

\section{Failure under GS-Induced Choice Correspondences}
\label[appendix]{app:pi-lad-failure}

The following example separates GS-representability of full firm preferences from GS rationalization of their choice correspondences.
Changing bundle comparisons that leave the choice correspondence unchanged can cause the procedure to return an unstable matching.
A Pareto-stable matching nevertheless exists in this market.

\begin{proposition}\label[proposition]{prop:pi-lad-failure}
Even when every firm's choice correspondence is induced by a GS function, generalized deferred acceptance can fail weak stability.
\end{proposition}

\begin{proof}
Let $I=\{a,b,c,d\}$ and $J=\{A,B\}$, with firm index order $A,B$.
We abbreviate worker sets as $ac=\{a,c\}$.
Workers $a$ and $d$ have preference $A\sim B\succ\varnothing$, worker $b$ has $B\succ A\succ\varnothing$, and worker $c$ has $A\succ B\succ\varnothing$.
Define the partition-matroid domains
\begin{align*}
 \cD_A&=\{X\subseteq I\mathrel{:} |X\cap\{a,d\}|\le1,\ |X\cap\{b,c\}|\le1\},\\
 \cD_B&=\{X\subseteq I\mathrel{:} |X\cap\{a,b\}|\le1\}.
\end{align*}
Define the firms' preferences using the utilities
\begin{align*}
 u_A(X)&=|X|+\mathbf{1}[b\in X] &&(X\in\cD_A),\\
 u_B(X)&=|X|+2\cdot\mathbf{1}[a\in X]+3\cdot\mathbf{1}[\{c,d\}\subseteq X] &&(X\in\cD_B),
\end{align*}
with value $-\infty$ outside the respective domains.
The function $u_A$ is GS, whereas the weak order induced by $u_B$ has no GS representation, as shown below.

Let $g_B(X)=|X|+2\cdot\mathbf{1}[a\in X]$ for $X\in\cD_B$, with value $-\infty$ otherwise.
Both $u_B$ and $g_B$ uniquely choose all available members of $\{c,d\}$, together with $a$ if available and otherwise $b$ if available.
Thus, $\cC_B(S)=\argmax\{u_B(X)\mathrel{:} X\subseteq S\}=\argmax\{g_B(X)\mathrel{:} X\subseteq S\}$ for every $S\subseteq I$.
Since $u_A$ and $g_B$ are additive on partition-matroid domains, they are GS functions inducing the firms' choice correspondences.

We now run generalized deferred acceptance in the firm order $A,B$.
\begin{itemize}
\item In the first round, $A$ retains $ac$ and $cd$.
Since $u_B(a)=3>2=u_B(bd)$, the unique retained matching is $\mu^{(1)}=\{(a,B),\allowbreak(c,A),\allowbreak(d,A)\}$, leaving $b$ unmatched.

\item In the second round, worker $b$ advances to $A$.
Firm $A$ retains $ab$ and $bd$, and $B$ prefers $a$ to $d$, yielding $\mu^{(2)}=\{(a,B),\allowbreak(b,A),\allowbreak(d,A)\}$ and leaving $c$ unmatched.

\item In the third round, worker $c$ advances to $B$.
Firm $A$ again retains $ab$ and $bd$, but now $u_B(cd)=5>4=u_B(ac)$.
The procedure terminates at $\mu^{(3)}=\{(a,A),\allowbreak(b,A),\allowbreak(c,B),\allowbreak(d,B)\}$.
\end{itemize}
Firm choices uniquely determine the matching in each round, establishing independence from the worker index order and the final selection rule.
Firm $B$ and $bcd$ block the terminal matching because $u_B(bcd)=6>5=u_B(cd)$ and worker $b$ strictly prefers $B$ to $A$.
\end{proof}

Although $g_B$ reproduces $B$'s choice correspondence, no GS function represents $B$'s full weak order.
Any such representation $v$ would satisfy the exchange axiom~\eqref{eq:exc} for $X=cd$, $Y=a$, and $e=c$, giving $v(cd)+v(a)\le\max\{v(d)+v(ac),\ v(ad)+v(c)\}$.
However, $cd\succ_B ac\sim_B ad$ and $a\succ_B c\sim_B d$ make both terms on the right strictly smaller than the left.

Nevertheless, $\widehat\mu=\{(a,B),\allowbreak(b,A),\allowbreak(c,B),\allowbreak(d,B)\}$ is Pareto-stable.
It is individually rational and admits no block, since $B$ receives its unique best bundle $acd$, workers $a,d$ are indifferent between the firms, and $A$ prefers $b$ to $c$.
Preserving $B$'s best bundle and $A$'s welfare fixes all assignments, ruling out a Pareto improvement.

\section{Strongly Polynomial Implementation}\label[appendix]{app:implementation}
We show that each round of generalized deferred acceptance amounts to one valuated matroid intersection.
Before workers report, compute the base $M$ of \Cref{sec:current-tier-clearing} with $nm$ singleton value queries, and fix $\kappa$ and the weights $\omega_i$ of \Cref{sec:worker-tiers}.
For a tier profile $F$, let $E_F=\{(i,j)\mid j\in F_i\}$, so that $\cM(F)$ consists of the matchings contained in $E_F$, and put $N=m+1$.
Let $a_i=1$ if worker $i$'s current tier lies strictly above $\varnothing$, and $a_i=0$ otherwise.

\begin{proposition}\label[proposition]{prop:single-vmi}
For every tier profile $F$, the matching $\Psi(F)$ is the unique maximizer over $\cM(F)$ of
\begin{align*}
 Q_F(\nu)=N^n\biggl[\kappa v(\nu)+\sum_{(i,j)\in\nu}a_i\omega_i\biggr]+\sum_{(i_q,j_t)\in\nu}(m+1-t)N^{n-q}.
\end{align*}
This maximizer can be computed by one valuated matroid intersection on $|E_F|+n$ elements.
Consequently, generalized deferred acceptance runs in strongly polynomial time given a value oracle for each GS representation.
\end{proposition}

\begin{proof}
The bracket is integer-valued.
The proof of \Cref{prop:ordered-welfare}~(i) applies directly, with the difference $\omega_i-z_i(F_i)$ replaced by $a_i\omega_i$.
Hence, the maximizers of the bracket over $\cM(F)$ are exactly the matchings in $\cM_{m+n}(F)$.
The last sum lies in $\{0,\ldots,N^n-1\}$, and its base-$N$ digits are the scores $m+1-t$ of the workers' firms $j_t$ in the worker index order, with score $0$ for an unmatched worker.
Maximizing $Q_F$ therefore retains $\cM_{m+n}(F)$ and then minimizes the assignment vector lexicographically under the order $j_1,\ldots,j_m,\varnothing$, which selects exactly $\Psi(F)$.

On subsets of $E_F$, the function $Q_F$ is the sum of a positive multiple of the GS function $v$ and an additive function, hence it is GS.
Add a dummy element $d_i$ for each worker $i$, and put $f_F(S)=Q_F(S\cap E_F)$ if $|S|=n$ and $f_F(S)=-\infty$ otherwise for $S\subseteq E_F\cup\{d_i\mid i\in I\}$.
By the equal-cardinality exchange property of GS functions~\cite[Remark~3.3 and (3.10)]{Murota2016}, $f_F$ is a valuated matroid of rank $n$.
The bases of the partition matroid in which each worker $i$ has the block $(E_F\cap(\{i\}\times J))\cup\{d_i\}$ correspond to the matchings $\nu\in\cM(F)$ through $\nu\mapsto\nu\cup\{d_i\mid\nu(i)=\varnothing\}$, and $f_F$ takes the value $Q_F(\nu)$ there.
Hence, $\Psi(F)$ is obtained by maximizing $f_F$ over the common bases, a valuated matroid intersection on $|E_F|+n$ elements.

It is known that valuated matroid intersection on a ground set $V$ with rank $r$ can be solved in $O(|V|r^2\gamma+|V|r\log|V|)$ time, where $\gamma$ is the time for one valuation query~\cite{Murota1996Algorithms}.
Here, $|V|=|E_F|+n\le(m+1)n$ and $r=n$, and each valuation query to $f_F$ uses at most $m$ value queries.
Each round therefore uses $O(n^3m^2)$ value queries and polynomially many arithmetic operations, and the at most $nm$ rounds use $O(n^4m^3)$ value queries in total.
The numbers $N^n$, $\kappa$, $\omega_i$, and $M^{m-t}$, as well as the values of $f_F$, have encoding lengths polynomial in $n$, $m$, and the encoding lengths of the finite oracle values.
Hence, the procedure runs in strongly polynomial time.
\end{proof}

\end{document}